\documentclass[a4paper,USenglish,cleveref,autoref,thm-restate]{lipics-v2021}

\usepackage{gensymb}

\newcommand{\eps}{\varepsilon}
\newcommand{\diam}{\text{diam}}

\newcommand{\dir}{\text{dir}}
\newcommand{\MST}{\text{MST}}
\newcommand{\opt}{\text{OPT}}
\newcommand{\alg}{\text{ALG}}
\newcommand{\N}{\mathbb{N}}
\newcommand{\R}{\mathbb{R}}
\newcommand{\Z}{\mathbb{Z}}
\newcommand{\sphere}{\mathbb{S}}

\hypersetup{
hidelinks,
colorlinks=true,
citecolor=[rgb]{0.121 0.47 0.705},
linkcolor=[rgb]{0.7 0 0.4},
urlcolor=[rgb]{0.121 0.47 0.705}
}

\hideLIPIcs
\nolinenumbers

\title{Online Euclidean Spanners}

\titlerunning{Online Euclidean Spanners}

\author{Sujoy Bhore}{Indian Institute of Science Education and Research, Bhopal, India.}{sujoy.bhore@gmail.com}{0000-0003-0104-1659}{}
\author{Csaba D. T\'oth}{California State University Northridge, Los Angeles, CA; and Tufts University, Medford, MA, USA.}{csaba.toth@csun.edu}{0000-0002-8769-3190}
{}
\authorrunning{S.~Bhore and C.~D.~T\'oth}

\Copyright{The authors}



\EventEditors{}
\EventNoEds{2}
\EventLongTitle{}
\EventShortTitle{}
\EventAcronym{}
\EventYear{}
\EventDate{}
\EventLocation{}
\EventLogo{}
\SeriesVolume{}
\ArticleNo{}

\begin{document}

\maketitle

\begin{abstract}
In this paper, we study the online Euclidean spanners problem for points in $\mathbb{R}^d$. Given a set $S$ of $n$ points in $\mathbb{R}^d$, a $t$-spanner on $S$ is a subgraph of the underlying complete graph $G=(S,\binom{S}{2})$, that preserves the pairwise Euclidean distances between points in $S$ to within a factor of $t$, that is the \emph{stretch factor}. Suppose we are given a sequence of $n$ points $(s_1,s_2,\ldots, s_n)$ in $\R^d$, where point $s_i$ is presented in step~$i$ for $i=1,\ldots, n$. The objective of an online algorithm is to maintain a geometric $t$-spanner on $S_i=\{s_1,\ldots, s_i\}$ for each step~$i$. The algorithm is allowed to \emph{add} new edges to the spanner when a new point is presented, but cannot \emph{remove} any edge from the spanner. The performance of an online algorithm is measured by its competitive ratio, which is the supremum, over all sequences of points, of the ratio between the weight of the spanner constructed by the algorithm and the weight of an optimum spanner. Here the weight of a spanner is the sum of all edge weights.

First, we establish a lower bound of $\Omega(\eps^{-1}\log n / \log \eps^{-1})$ for the competitive ratio of any online $(1+\eps)$-spanner algorithm, for a sequence of $n$ points in 1-dimension. We show that this bound is tight, and there is an online algorithm that can maintain a $(1+\eps)$-spanner with competitive ratio $O(\eps^{-1}\log n / \log \eps^{-1})$. Next, we design online algorithms for sequences of points in $\R^d$, for any constant $d\ge 2$, under the $L_2$ norm.
We show that previously known incremental algorithms achieve a competitive ratio $O(\eps^{-(d+1)}\log n)$. However, if the algorithm is allowed to use additional points (Steiner points), then it is possible to substantially improve the competitive ratio in terms of $\eps$. We describe an online Steiner $(1+\eps)$-spanner algorithm with competitive ratio $O(\eps^{(1-d)/2} \log n)$.
As a counterpart, we show that the dependence on $n$ cannot be eliminated in dimensions $d \ge 2$. In particular, we prove that any online spanner algorithm for a sequence of $n$ points in $\mathbb{R}^d$ under the $L_2$ norm has competitive ratio $\Omega(f(n))$, where
$\lim_{n\rightarrow \infty}f(n)=\infty$.
Finally, we provide improved lower bounds under the $L_1$ norm: $\Omega(\eps^{-2}/\log \eps^{-1})$ in the plane and $\Omega(\eps^{-d})$ in $\mathbb{R}^d$ for $d\geq 3$.

\end{abstract}



\section{Introduction}\label{sec:intro}
We study the online Euclidean spanners problem for a set of points in $\mathbb{R}^d$. Let $S$ be a set of $n$ points in $\mathbb{R}^d$. A $t$-spanner for a finite set $S$ of points in $\mathbb{R}^d$ is a subgraph of the underlying complete graph $G=(S,\binom{S}{2})$, that preserves the pairwise Euclidean distances between points in $S$ to within a factor of $t$, that is the \emph{stretch factor}. The edge weights of $G$ are the Euclidean distances between the vertices. Chew~\cite{Chew86, Chew89} initiated the study of Euclidean spanners in 1986, and showed that for a set of $n$ points in $\mathbb{R}^2$, there exists a spanner with $O(n)$ edges and constant stretch factor. Since then a large body of research has been devoted to Euclidean spanners due to its vast applications across domains, such as, topology control in wireless networks~\cite{schindelhauer2007geometric},
efficient regression in metric spaces~\cite{gottlieb2017efficient},
approximate distance oracles~\cite{gudmundsson2008approximate}, and many others. Moreover, Rao and Smith~\cite{rao1998approximating} showed the relevance of Euclidean spanners in the context of other fundamental geometric \textsf{NP}-hard problems, e.g., Euclidean traveling salesman problem and Euclidean minimum Steiner tree problem. Many different spanner construction approaches have been developed for Euclidean spanners over the years, that each found further applications in geometric optimization, such as spanners based on
well-separated pair decomposition (WSPD)~\cite{callahan1993optimal, GudmundssonLN02},
skip-lists~\cite{arya1994randomized},
path-greedy and gap-greedy approaches~\cite{althofer1993sparse, arya1997efficient}, locality-sensitive orderings~\cite{ChanHJ20}, and more. We refer to the book by Narasimhan and Smid~\cite{narasimhan2007geometric} and the survey of Bose and Smid~\cite{BoseS13} for a summary of results and techniques on Euclidean spanners up to 2013.

\smallskip\noindent\textbf{Online Spanners.} We are given a sequence of $n$ points $(s_1,s_2,\ldots , s_n)$, where the points are presented one-by-one, i.e., point $s_i$ is revealed at the step~$i$, and $S_i=\{s_1,\ldots ,  s_i\}$ for $i=1,\ldots ,n$. The objective of an online algorithm is to maintain a geometric $t$-spanner $G_i$ for $S_i$ for all $i$. Importantly, the algorithm is allowed to \emph{add} edges to the spanner when a new point arrives, however is not allowed to \emph{remove} any edge from the spanner.

The performance of an online algorithm $\alg$ is measured by comparing it to the offline optimum $\opt$ using the standard notion of competitive ratio~\cite[Ch.~1]{BY98}.
The \emph{competitive ratio} of an online $t$-spanner algorithm $\alg$ is defined as $\sup_\sigma \frac{\alg(\sigma)}{\opt(\sigma)}$,
where the supremum is taken over all input sequences $\sigma$, $\opt(\sigma)$ is the minimum weight of a $t$-spanner for $\sigma$, and $\alg(\sigma)$ denotes the weight of the $t$-spanner produced by $\alg$ for this input.


Computing a $(1+\eps)$-spanner of minimum weight for a set $S$ in Euclidean plane
is known to be \textsf{NP}-hard~\cite{carmi2013minimum}.
However, there exists a plethora of constant-factor
approximation algorithms for this problem in the offline model; see~\cite{althofer1993sparse,das1993optimally, narasimhan1995new,rao1998approximating}. Most of these algorithms approximate the parameter
\emph{lightness} (the ratio of the spanner weight to the weight of the Euclidean minimum spanning tree $\MST(S)$) of Euclidean spanners, which in turn also approximates the optimum weight of the spanner.
%
%
We refer to Section~\ref{intro-relatedwork} for a more detailed overview of the parameter lightness.

\emph{Minimum spanning trees (\MST)} on $n$ points in a metric space, which have no guarantee on the stretch factor, have been studied in the online model. It is not difficult to show that a greedy algorithm achieves a competitive ratio $\Theta(\log n)$. The online Steiner tree problem was studied by Imase and Waxman~\cite{imase1991dynamic}, who proved $\Theta(\log n)$-competitiveness for the problem. Later, Alon and Azar~\cite{AlonA93} studied minimum Steiner trees for points in the Euclidean plane, and proved a lower bound $\Omega(\log n/\log \log n)$ for the competitive ratio. Their result was the first to analyse the impact of Steiner points on a geometric network problem in the online setting. Several algorithms were proposed over the years for the online Steiner Tree and Steiner forest problems, on graphs in both weighted and unweighted settings; see~\cite{alon2006general, awerbuch2004line,  berman1997line, hajiaghayi2013online, naor2011online}.

\smallskip\noindent\textbf{Online Steiner Spanners.} An important variant of online spanners is when it is allowed to use auxiliary points (Steiner points) which are not part of input sequence of points. It turns out that Steiner points allow for substantial improvements over the bounds on the sparsity and lightness of Euclidean spanners in the offline settings; see~\cite{BT-lessp-21, BT-oess-21, le2019truly, le2020light}.
In the geometric setting, an online algorithm is allowed to \emph{add} Steiner points and \emph{subdivide} existing edges with Steiner points at each time step.
(This modeling decision has twofold justification: It accurately models physical networks such as roads, canals, or power lines, and from the theoretical perspective, it is hard to tell whether an online algorithm introduced a large number of Steiner points when it created an edge/path in the first place).
However, the spanner must achieve the given stretch factor only for the input point pairs.

It is easy to see that in this model the online spanners in $1$-dimension could attain optimum competitive ratio. However, it is unclear how it extends to higher dimensions as it has been observed in the offline settings that it tends to be more difficult to achieve tight bounds for Steiner spanners than their non-Steiner counterparts.
%

When the optimal Steiner spanner is lighter than $\opt(S_i)$ without Steiner points, the adversary may decrease $\opt(S_i)$ by adding suitable Steiner vertices to $S_i$; see Fig.~\ref{fig:online-intro}. In particular, $\opt(S_i)$ may or may not increase with $i$ in the model without Steiner points, but $\opt(S_i)$ monotonically increases in $i$ when Steiner points are allowed.

\begin{figure}
    \centering
    \includegraphics[width=0.5\textwidth]{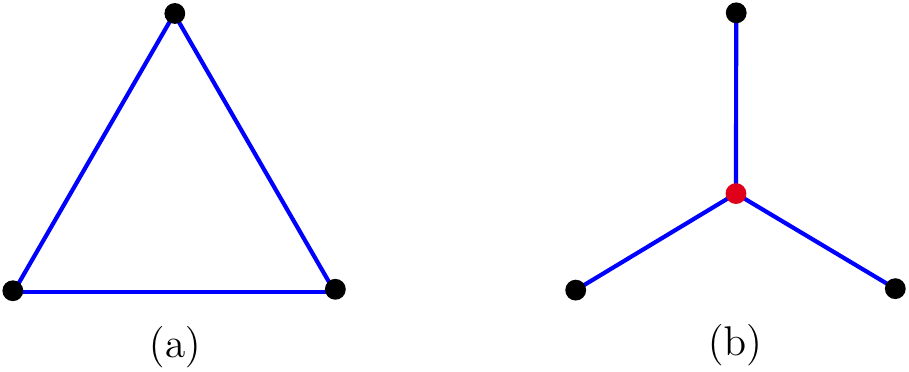}
    \caption{(a) An optimum $\frac32$-spanner on three points with all edges of unit length.
    (b) After inserting a fourth point at the center, the weight of the optimum $\frac32$-spanner decreases.}
    \label{fig:online-intro}
\end{figure}

\subsection{Related Work}\label{intro-relatedwork}

\smallskip\noindent\textbf{Dynamic Spanners.}
In applications, the data (modeled as points in $\mathbb{R}^d$) changes over time, as new cities emerge, new wireless antennas are built, and users turn their wireless devices on or off. \emph{Dynamic} models aim to maintain a geometric $t$-spanners for a dynamically changing point set $S$; in a restricted \emph{insert-only} model, the input consists of a sequence of point insertions.
In the dynamic model, the objective is design algorithms and data structures that minimize the worst-case update time needed to maintain a $t$-spanner for $S$ over all steps, regardless of its weight, sparsity, or lightness. Notice that dynamic algorithms are allowed to add or delete edges in each step, while online algorithms cannot delete edges. However, if a dynamic (or dynamic insert-only) algorithm always adds edges for a sequence of points insertions, it is also an online algorithm, and one can analyze its competitive ratio.

Arya et al.~\cite{arya1994randomized} designed a randomized incremental algorithm for $n$ points in $\R^d$, where the points are inserted in a random order, and maintains a $t$-spanner of $O(n)$ size and $O(\log n)$ diameter. Their algorithm can also handle random insertions and deletions in $O(\log^d n \log \log n)$ expected amortized update time.
%
Later, Bose et al.~\cite{bose2004ordered} presented an insert-only algorithm
to maintain a $t$-spanner of $O(n)$ size and $O(\log n)$ diameter in $\R^d$.
Fischer and Har-Peled~\cite{FischerH05} used dynamic compressed quadtrees to maintain a WSPD-based $(1+\eps)$-spanner for $n$ points in $\R^d$ in expected $O([\log n+\log \eps^{-1}]\,\eps^{-d}\log n)$ update time. Their algorithm works under the online model, too, however, they have not analyzed the weight of the resulting spanner.
Gao et al.~\cite{gao2006deformable} used hierarchical clustering for dynamic
spanners in $\R^d$. Their {\sc DefSpanner} algorithm is fully dynamic with $O(\log \Delta)$ update time, where $\Delta$ is the spread\footnote{The \emph{spread} of a finite set $S$ in a metric space is the ratio of the maximum pairwise distance to the minimum pairwise distance of points in $S$; and $\log \Delta\geq \Omega(\log n)$ in doubling dimensions.} of the set $S$.
They maintain a $(1+\eps)$-spanner of weight $O(\eps^{-(d+1)}\|MST(S)\|\log \Delta)$, and for a sequence of point insertions, {\sc DefSpanner} only adds edges. As $\opt\geq \|MST(S)\|$,  {\sc DefSpanner} can serve as an online algorithm with competitive ratio $O(\eps^{-(d+1)}\log \Delta)$.

Gottlieb and Roditty~\cite{gottlieb2008optimal} studied dynamic spanners in more general settings. For every set of $n$ points in a metric space of bounded doubling dimension\footnote{A metric is said to be of a \emph{constant doubling dimension} if a ball with radius $r$ can be covered by at most a constant number of balls of radius $r/2$.}, they constructed a $(1+\eps)$-spanner whose maximum degree is
$O(1)$ and that can be maintained under insertions and deletions in $O(\log n)$ amortized update time per operation.
Later, Roditty~\cite{Roditty12} designed fully dynamic geometric $t$-spanners with optimal $O(\log n)$ update time for $n$ points in $\mathbb{R}^d$. Very recently, Chan et al.~\cite{ChanHJ20} introduced \emph{locality sensitive orderings} in $\R^d$, which has applications in several proximity problems, including spanners. They obtained a fully dynamic data structure for maintaining a
$(1 + \eps)$-spanners in Euclidean space with logarithmic update time and linearly many edges. However, the spanner weight has not been analyzed for any of these constructions. Dynamic spanners have been subject to investigation in abstract graphs, as well. See~\cite{BaswanaKS12, BergamaschiHGWW21, BernsteinFH19} for some recent progress on dynamic graph spanners.

\smallskip\noindent\textbf{Lightness and sparsity} are two natural parameters for Euclidean spanners.
For a set $S$ of points in $\mathbb{R}^d$, the lightness is the ratio of the spanner weight (i.e., the sum of all edge weights) to the weight of the Euclidean minimum spanning tree $MST(S)$.
It is known that \emph{greedy-spanner} (\cite{althofer1993sparse}) has constant lightness; see~\cite{das1993optimally, narasimhan1995new}. Later, Rao and Smith~\cite{rao1998approximating} in their seminal work, showed that the greedy spanner has lightness $\eps^{-O(d)}$ in $\mathbb{R}^d$ for every constant $d$, and asked what is the best possible constant in the exponent.
Then, the \emph{sparsity} of a spanner on $S$ is the ratio of its size to the size of a spanning tree. Classical results~\cite{Chew89, Clarkson87, keil1988approximating, yao1982constructing}
show that when the dimension $d\in \mathbb{N}$ and $\varepsilon > 0$ are
constant, every set $S$ of $n$ points in $d$-space admits an $(1 +\varepsilon)$-spanners with $O(n)$ edges and weight proportional to that of the Euclidean \textsc{MST} of $S$.

\smallskip\noindent\textbf{Dependence on $\eps>0$ for constant dimension $d$.}
The dependence of the lightness and sparsity on $\eps>0$ for constant $d\in \mathbb{N}$ has been studied only recently. Le and Solomon~\cite{le2019truly} constructed, for every $\eps>0$ and constant $d\in \mathbb{N}$, a set $S$ of $n$ points in $\mathbb{R}^d$ for which any $(1+\eps)$-spanner must have lightness $\Omega(\eps^{-d})$ and sparsity $\Omega(\eps^{-d+1})$, whenever $\eps = \Omega(n^{-1/(d-1)})$. Moreover, they showed that the greedy $(1+\eps)$-spanner in $\mathbb{R}^d$ has lightness $O(\eps^{-d}\log \eps^{-1})$. In fact, Le and Solomon~\cite{le2019truly} noticed that Steiner points can substantially improve the bound on the lightness and sparsity of an $(1+\eps)$-spanner.
For minimum sparsity, they gave an upper bound of $O(\eps^{(1-d)/2})$ for $d$-space and a lower bound of $\Omega(\eps^{-1/2}/\log\eps^{-1})$. For minimum lightness, they gave a lower bound of $\Omega(\eps^{-1}/\log\eps^{-1})$, for points in the plane ($d=2$)~\cite{le2019truly}. More recently, Bhore and T\'{o}th~\cite{BT-oess-21} established a lower bound of $\Omega(\eps^{-d/2})$ for the lightness of Steiner $(1+\eps)$-spanners in Euclidean $d$-space for all $d\ge 2$. Moreover, for points in the plane, they established an upper bound of $O(\eps^{-1})$~\cite{BT-lessp-21}.

\subsection{Our Contributions}

We present the main contributions of this paper, and sketch the key technical and conceptual ideas used for establishing these results. (Refer to the technical sections for precise definitions, complete proofs, and additional remarks.)

\smallskip\noindent\textbf{Points on a line.}
In Section~\ref{sec:1-D} (Theorem~\ref{thm:1D-bounds}),
we establish a lower bound $\Omega(\eps^{-1}\log n/\log \eps^{-1})$
for the competitive ratio of any online algorithm for a sequence of points on the real line. Moreover, we show that this bound is tight. We present an online algorithm that  maintains a $(1+\eps)$-spanner with competitive ratio $O(\eps^{-1}\log n/\log \eps^{-1})$.

Our online algorithm is a 1-dimensional instantiation of hierarchical clustering, which was used by Roditty~\cite{Roditty12} for dynamical spanners in doubling metrics. When a new point $s_i$ is ``close'' to a previous point $s_j$, we add $s_i$ to the ``cluster'' of $s_j$, otherwise we open a new cluster. The key question is to define when $s_i$ is ``close'' to a previous point. Instead of the closest points on the line, we find the shortest edge $pq$ that contains $s_i$ in the current spanner, and say that $s_i$ is ``close'' to $p$ (resp., $q$) if $\|ps_i\|\leq \frac{\eps}{4}\|pq\|$ (resp., $\|qs_i\|\leq \frac{\eps}{4}\|pq\|$). The algorithm (and its analysis), does not explicitly maintain ``clusters,'' though. It is easy to show, by induction, that $\alg$ maintains a $(1+\eps)$-spanner. The main contribution is a tight analysis of the competitive ratio. We partition the edges into \emph{buckets} by weight, where bucket $E_\ell$ contains edges $e$ of weight $\eps^{-(\ell+1)}<\|e\|\leq\eps^{-\ell}$. The edges of the spanner will form a laminar family (any edges are interior-disjoint or one contains the other); and the edge weight decay by factors of at most $(1-\frac{\eps}{4})$ along the descending paths in the containment poset.
Since $(1-\frac{\eps}{4})^{4/\eps}<\frac12$, we can show that the total weight of edges in a level decreases by a factor of $\frac12$ after every $\lceil 5/\eps\rceil$ levels. Thus, the sum of edge weights in a \emph{block} of $\lceil 5/\eps\rceil$ consecutive levels is $O(\eps^{-1}\opt)$. This bound, applied to $O(\log_{\eps^{-1}} n)=O(\log n/\log \eps^{-1})$ buckets,  proves the upper bound.
The lower bound construction matches the upper bound for each block of levels and for each bucket.

\smallskip\noindent\textbf{Euclidean $d$-space without Steiner points.}
In Section~\ref{sec:L2-upper}, we study the online Euclidean spanners for a sequence of points in $\mathbb{R}^d$. For constant $d\ge 2$ and parameter $\eps > 0$, we show that the dynamic algorithm by Fischer and Har-Peled achieves, in the online model, competitive ratio $O(\eps^{-(d+1)}\log n)$  for $n$ points in $\mathbb{R}^d$ (Theorem~\ref{thm:L2woSteiner} in Section~\ref{sec:L2-upper-woSteiner}), matching the competitive ratio of {\sc DefSpanner} by Gao et al.~\cite[Lemma~3.8]{gao2006deformable}.

The new competitive analysis of this algorithm is instrumental for extending the algorithm and its analysis to online Steiner $(1+\eps)$-spanners (see below). We briefly describe a key geometric insight. It is well known that for $a,b\in \R^d$, any $ab$-path of weight at most $(1+\eps)\|ab\|$ lies in an ellipsoid $B_{ab}$ with foci $a$ ans $b$ and great axes $(1+\eps)\|ab\|$. Summation over \emph{disjoint} ellipses gives a lower bound for $\opt$. Unfortunately, ellipsoids $B_{ab}$ for all pairs $ab\in S$ may heavily overlap.
Recently, Bhore and T\'oth~\cite[Lemma~3]{BT-oess-21} proved that any $ab$-path of weight at most $(1+\eps)\|ab\|$ must contain edge of total weight at least $\frac12 \|ab\|$ that are ``near-parallel'' to $ab$ (technically, they make an angle at most $\eps^{1/2}$ with $ab$); see Fig.~\ref{fig:quadtree}(right). By partitioning the edges of the unknown $\opt$ spanner by \emph{both} directions and disjoint ellipsoids, we obtain a bound of $\frac{\alg}{\opt}\leq O(\eps^{-(d+1)}\log n)$.

\smallskip\noindent\textbf{Euclidean $d$-space with Steiner points.}
When we are allowed to use Steiner points, we can substantially improve the competitive ratio in terms of $\eps$: We describe an algorithm with competitive ratio
$O(\eps^{(1-d)/2} \log n)$ (Theorem~\ref{thm:L2withSteiner2D} in Section~\ref{sec:L2-upper-withSteiner}).

The online Steiner algorithm adds a secondary layer to the non-Steiner algorithm: For each edge $ab$ of the non-Steiner spanner $G_1$, we maintain a path of weight $(1+\eps)\|ab\|$ with Steiner points; the stretch factor of the resulting Steiner spanner $G_2$ is $(1+\eps)^2<(1+3\eps)$. The key idea is to reduce the weight to maintain \emph{buckets} of edges of $G_1$ that have roughly the same direction and  weight, and are nearby locations; and we construct a common Steiner network $N$ for them. Importantly, we can construct a ``backbone'' of the network $N$ when the first edge $ab$ in a bucket arrives, and we have $\|N\|\leq O(\eps^{(1-d)/2}\|ab\|)$. When subsequent edges $a'b'$ in the same bucket arrive, then we can add relatively short ``connectors'' to $N$ so that it also contains an $a'b'$-path of weight at most $(1+\eps)\|a'b'\|$. Thus $N$ can easily accommodate new paths in the online model. The key technical tool for constructing Steiner networks $N$ (one for each bucket) is the so-called \emph{shallow-light trees}, introduced by Awerbuch et al.~\cite{AwerbuchBP90} and Khuller et al.~\cite{KhullerRY93}, and optimized in the geometric setting by Elkin and Solomon~\cite{elkin2015steiner,Solomon15}.

As a counterpart, we show (Theorem~\ref{thm:w-SP-LB} in Section~\ref{sec:L2-lower-withSteiner}) that the dependence on $n$ cannot be eliminated in dimensions $d \ge 2$. In particular, we prove that any $(1+\eps)$-spanner for a sequence of $n$ points in $\mathbb{R}^d$, has competitive ratio $\Omega(f(n))$ for some function $f(n)$ with
$\lim_{n\rightarrow \infty}f(n)=\infty$.
The lower bound construction consists of an adaptive strategy for the adversary in the plane: The adversary recursively maintains a space partition and places points in \emph{rounds} so that the spanner constructed so far is disjoint from most of the ellipses $B_{ab}$ that will contains the $ab$-paths for pairs of new points $a,b$. In order to control $\opt$, the adversary maintains the property that $\opt_i$ is an $x$-monotone path $\gamma_i$ after round $i$. However, this requirement means that any new point must be very close to $\gamma_i$, and $S$ will be a set of almost collinear points. The core challenge of the Steiner spanner problem seems to lie in the case of almost collinear points.

\smallskip\noindent\textbf{Higher dimensions under the $L_1$-norm.}
Finally, in Section~\ref{sec:lower-bound-L1} we provide improved lower bounds for points in $\mathbb{R}^d$ under the $L_1$ norm (without Steiner points). We show that for every $\eps>0$, under the $L_1$ norm, the competitive ratio of any online $(1+\eps)$-spanner algorithm $\Omega(\eps^{-2}/\log \eps^{-1})$ in $\R^2$ and is $\Omega(\eps^{-d})$ in $\R^d$ for $d\geq 3$.

The adversary takes advantage of the non-monotonicity of $\opt$, mentioned above. In round~1, it presents a point set $S_1\cup S_2$ for which any $(1+\eps)$-spanner (without Steiner points) must contain a complete bipartite graph between  $S_1$ and $S_2$; however the optimal Steiner $(1+\eps)$-spanner for $S_1\cup S_2$ has much smaller weight. Then in round~2, the adversary presents all Steiner points $\widehat{S}_1\cup \widehat{S}_2$ of an optimal Steiner $(1+\eps)$-spanner for $S_1\cup S_2$. The key insight is that under the $L_1$-norm (and for this particular point set), the optimal Steiner spanner for $S_1\cup S_2$ already contains Manhattan paths between any two points in $S=(S_1\cup S_2)\cup (\widehat{S}_1\cup \widehat{S}_2)$, and so it remains the optimum solution (without Steiner points) for the point set $S$.

We were unable to replicate this phenomenon under the $L_2$-norm, where the current best lower bound in $\R^d$, for all $d\geq 1$, derives from the 1-dimensional construction.
In particular, it is not sufficient to consider the \emph{Steiner ratio for $(1+\eps)$-spanners}, defined as the supremum ratio between the weight of the minimum $(1+\eps)$-spanner and the minimum Steiner $(1+\eps)$-spanner of a finite point set in $\R^d$. Under the $L_2$-norm, this ratio is $\Theta(\eps^{-1})$ in the plane and $\tilde{\Theta}(\eps^{(1-d)/2})$ in $\R^d$
for $d\geq 3$~\cite{BT-lessp-21,le2019truly,le2020unified}. However, an optimal Steiner $(1+\eps)$-spanner, need not achieve the desired $1+\eps$ stretch factor for the Steiner points.


\section{Lower and Upper Bounds for Points on a Line}\label{sec:1-D}

It is easy to analyze the one-dimensional case as the offline optimum network (\opt) for any set of points in a line is a path from the leftmost point to the rightmost point; the stretch factor of this path is always 1. (In contrast, in 2- and higher dimensions, the optimum $(1+\eps)$-spanner is highly dependent on the distribution of points, which in turn may change over time in the online model.)

\begin{figure}[htbp]
	\centering
 \includegraphics[width=\textwidth]{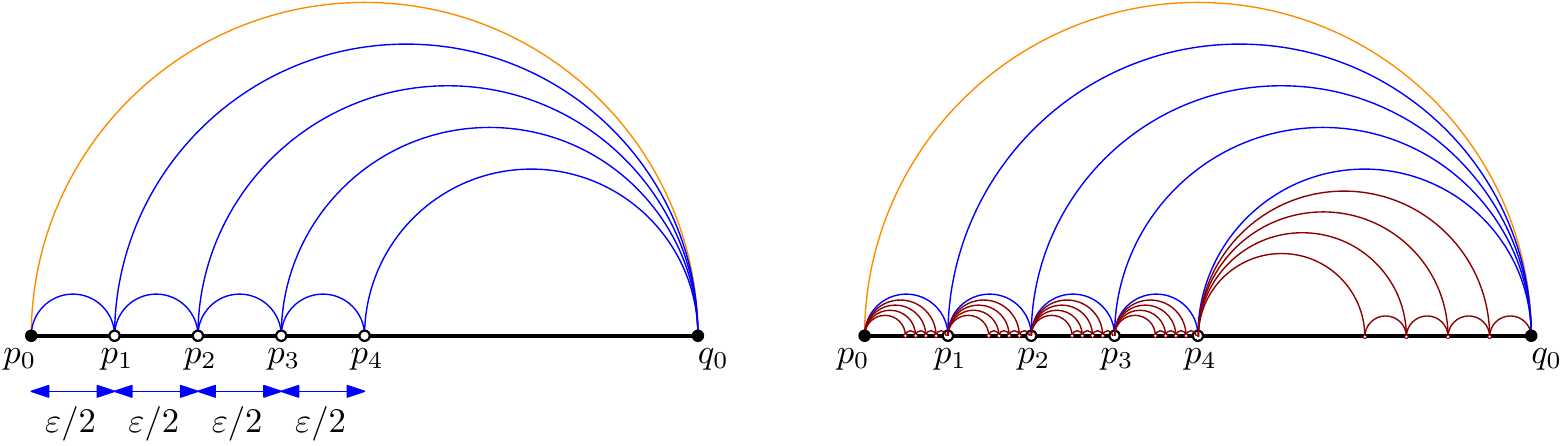}
\caption{Left: A sequence of $n$ points $(q_0,p_0,p_1,p_2,\ldots , p_{n-2})$, for $n=\lceil \eps^{-1}\rceil$, for which any online $(1+\eps)$-spanner has weight $\Omega(\eps^{-1}\opt)$. For clarity, the edges are drawn as circular arcs, but the weight of an edge $p_ip_j$ is $\|p_ip_j\|=|p_i-p_j|$.
Right: Iteration in each subinterval.}
\label{fig:Lower1D}
\end{figure}

\subparagraph{Lower bound.}
The following adversarial strategy establishes a lower bound $L(n)=\Omega(\eps^{-1})$ for the competitive ratio; refer to Fig.~\ref{fig:Lower1D}~(left). Start with two points $p_0=0$ and $q_0=1$. For the first two points, $\alg$ must add a direct edge $p_0q_0$. Then the adversary successively places points $p_i=i\cdot \frac{\eps}{2}$, for $i=1,\ldots ,n$ so that all points remain in the interval $[0,\frac12]$. Thus the number of points is $n=2+\lfloor \eps^{-1}\rfloor$.
In each round, $\alg$ must add the edge $p_iq_0$, otherwise any path between $p_i$ and $q_0$ would have to make a detour via a point in $\{p_0,\ldots ,p_{i-1}\}$, and so it would be longer than $(1+\eps)\|p_iq_0\|$. Since $\|p_iq_0\|\geq \frac12$, the weight of the network after $n-2$ iterations is at least $\alg\geq 1+\frac12(n-2)\geq 1+\frac12 \lfloor \eps^{-1}\rfloor$. Combined with $\opt=1$, this yields a lower bound of $\Omega(\eps^{-1})$ for the competitive ratio.

The adversary has placed only $O(\eps^{-1})$ points so far; this is the first stage of the strategy. In subsequent stages, the adversary repeats the same strategy in every subinterval $ab$ of previous stage, as indicated in Fig.~\ref{fig:Lower1D}~(right). After stage $j\geq 1$, we have $\alg\geq 1+\frac{j}{2} \lfloor \eps^{-1}\rfloor=\Omega(j\eps^{-1})$ and $\opt=1$.
The number of points placed in each stage increases by a factor of $\Omega(\eps^{-1})$, hence $j=\Theta(\log_{\eps^{-1}}n)=\Theta(\log n/\log \eps^{-1})$. Overall, the competitive ratio is at least $\alg/\opt\geq \Omega(j\eps^{-1})= \Omega(\eps^{-1}\log n / \log \eps^{-1})$.

\subparagraph{Upper bound.}
For proving a matching upper bound in one-dimension, we use the following online algorithm: For all $i=1,\ldots ,n$, we maintain a spanning graph $G_i$ on $S_i=\{s_1,\ldots , s_i\}$ and the $x$-monotone path $P_i$ between the leftmost and the rightmost points in $S_i=\{s_1,\ldots , s_i\}$. When point $s_i$, $i\geq 2$, arrives, we proceed as follows (see Fig.~\ref{fig:Upper1D}). If $s_i$ is left (resp., right) of all previous points, we add an edge from $s_i$ to the closest point in $S_{i-1}$ to both $P_{i-1}$ and $G_{i-1}$. Otherwise, let $ab$ be the (unique) edge of $P_{i-1}$ that contains $s_i$, and $pq$ a shortest edge of $G_{i-1}$ that contains $s_i$. Clearly, we have $P_i=P_{i-1}-ab+a s_i+s_i b$.
If $\min\{\|p s_i\|,\|s_i q\|\}> \frac{\eps}{4}\,\|pq\|$,
we add both $as_i$ and $s_ib$ to $G_i$, that is, $G_i=G_{i-1}+ a s_i+s_i b$.
Otherwise, let $G_i=G_{i-1}+ as_i$ if $\|p s_i\|\leq \|s_i q\|$,
or else $G_i=G_{i-1}+s_ib$.

\begin{figure}[htbp]
	\centering
 \includegraphics[width=\textwidth]{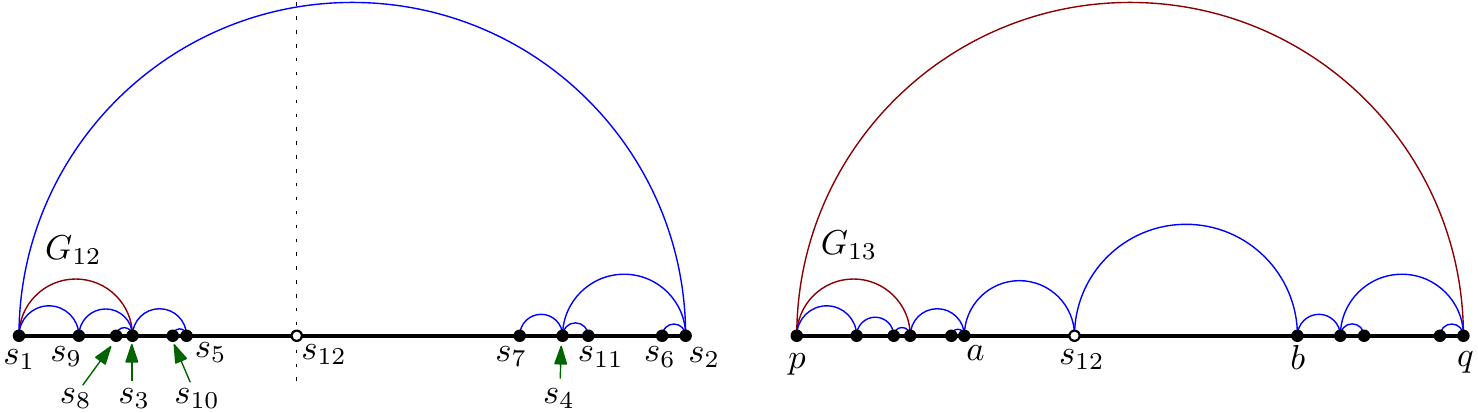}
\caption{Left: The graph $G_{12}$ for $(s_1,\ldots , s_{11})$.
Right: $pq=s_1s_2$ is the shortest edge of $G_{11}$ that contain $s_{12}$.
The algorithm adds edges $as_{12}=s_5s_{12}$ and $s_{12}b=s_{12}s_7$.}
\label{fig:Upper1D}
\end{figure}

We observe a few properties of $G_i$ that are immediate from the construction: (P1) At the time when edge $e$ is added to $G_i$, then the interior of $e$ does not contain any vertices. (P2) The edges in $G_i$ form a laminar set of intervals (i.e., any two edges are interior-disjoint, or one contains the other).
(P3) If $e_1,e_2$ are edges in $G_i$ and $e_2\subset e_1$, then $\|e_2\|\leq (1-\frac{\eps}{4})\|e_1\|$.
We note that properties (P1)--(P3) are inherently 1-dimensional, as the edges are intervals in $\R$, and they do not seem to generalize to higher dimensions.

\begin{lemma}\label{lem:1D}
For $i=1,\ldots ,n$, the graph $G_i$ is a $(1+\eps)$-spanner for $S_i$.
\end{lemma}

\begin{proof}
We proceed by induction on $i$. The base case $i=1$ is trivial. Assume that $i>1$ and $G_{i-1}$ is a $(1+\eps)$-spanner for $S_{i-1}$. It is enough to show that for every $j=1,\ldots , i-1$, $G_i$ contains an $s_is_j$-path of weight at most $(1+\eps)\|s_is_j\|$. We consider all cases step $i$ of the algorithm.

\noindent \textbf{Case~1:} Assume that $s_i$ is left of (resp., right of) $S_{i-1}$, and the closest previous point is $a\in S_{i-1}$. By induction, $G_{i-1}$ contains a path $\pi(a,s_j)$ of weight at most $(1+\eps)\|a s_j\|$. Hence the weight of the path $s_ia+\pi(a,s_j)$ is at most $\|s_ia\|+(1+\eps)\|a s_j\|< (1+\eps)(\|s_ia\|+\|a s_j\|=(1+\eps)\|s_i s_j\|$.

\noindent \textbf{Case~2:} Assume that $pq$ is the shortest edge of $G_{i-1}$ that contains $s_i$, and $a,b\in S_{i-1}$ are the closest previous points to $s_i$ on the left and right, respectively.
If $\min\{\|ps_{i}\|,\|s_{i}q\|\}>\frac{\eps}{4}\,\|pq\|$, then $G_i=G_{i-1} + a s_i+s_i b$, and we can argue similarly to Case~1.

Otherwise, we have $\min\{\|ps_{i}\|,\|s_{i}q\|\}\leq \frac{\eps}{4}\,\|pq\|$. Assume w.l.o.g. that $\|ps_{i}\|\leq \|s_{i}q\|$, consequently $G_i=G_{i-1}+a s_i$.
If $s_j$ is to the left of $s_i$, we can argue similarly to Case~1.
Hence we may assume that $s_j$ is to the right of $s_i$.

Property (P2) implies that $G_{i-1}$ contains an $x$-monotone $ap$-path $\pi(a,p)$, which has optimal weight $\|ap\|$; and an $x$-monotone $qb$-path $\pi(q,b)$ of weight $\|bq\|$.
Property (P1) implies that $a=p$ or $a$ was inserted after $p$ and $q$; and since $a$ is not the left endpoint of any edge in $G_i$, then
$\|ap\|\leq \frac{\eps}{4}\|pq\|$.
Analogously, we have $\|bq\|\leq \frac{\eps}{4}\|pq\|$. Consequently,
\[
\|s_ib\|
=\|pq\|-\|ps_i\|-\|bq\|
\geq \left(1-2\cdot \frac{\eps}{4}\right)\|pq\|
= \left(1-\frac{\eps}{2}\right)\|pq\|.
\]
Now we can construct an $s_ib$-path $\pi(s_i,b)=s_ia+\pi(a,p)+pq+\pi(q,b)$ of weight at most
\[
\|\pi(s_i,b)\|
\leq \|s_i p\|+\|pq\| +\|qb\|
\leq \left(1+\frac{\eps}{2}\right)\|pq\|
\leq \frac{1+\eps/2}{1-\eps/2}\,\|s_i b\|
<(1+\eps)\|s_i b\|.
\]
We can complete the proof now. By induction, $G_{i-1}$ contains a $bs_j$-path $\pi(b,s_j)$ of weight at most $(1+\eps)\|b s_j\|$. Hence the weight of the path $\pi(s_i,b)+\pi(b,s_j)$ is at most $(1+\eps)\|s_ib\|+(1+\eps)\|b s_j\|= (1+\eps)(\|s_ib\|+\|b s_j\|=(1+\eps)\|s_i s_j\|$, as required.
\end{proof}

\begin{lemma}\label{lem:1D-weight}
For $i=1,\ldots ,n$, we have $\|G_i\|\leq O(\eps^{-1}\opt_i\log i / \log \eps^{-1})$.
\end{lemma}
\begin{proof}
We may assume w.l.o.g.\ that $i=n$, and let $\opt=\opt_n$ for brevity.
Let $E$ be the edge set of $G_n$. The order in which $\alg$ adds edges to $E$ defines a (precedence) poset on $E$. We partition $E$ by weight as follows:
Let $\beta=\eps^{-1}$; and for all $\ell\in \Z$, let $E_\ell$ be the set of edges $e\in E$ with $\beta^{\ell}<\|e\|\leq \beta^{\ell+1}$. Since $\|e\|\leq \opt$ for all $e\in E$, every edge is in $E_\ell$ for some $\ell\leq \log_\beta \opt$. Furthermore, for all $\ell\leq  \log_\beta (\opt/n^2)$, the edges $e\in E_\ell$ have weight $\|e\|\leq \opt/n^2$, and so the total weight of these edges is less than $\opt$. It remains to consider $E_\ell$ for $\log_\beta (\opt/n^2)\leq \ell\leq \log_\beta \opt$, that is, for $O(\log n/\log \eps^{-1})$ values of $\ell$.

Let $pq$ be an edge in $E_\ell$ that is not contained in any previous edge in $E_\ell$. By property (P2), the edges in $E_\ell$ form a laminar family, and so $pq$ does not overlap with any previous edge in $E_\ell$; and $pq$ contains any subsequent edge that overlaps with it. Let $E_\ell(pq)$ be the set of all edges in $E_\ell$ that are contained in $pq$ (including $pq$). We claim that
\begin{equation}\label{eq:1}
    \|E_\ell(pq)\| \leq O(\eps^{-1}\|pq\|).
\end{equation}
Summation over all edges $pq\in E_\ell$ that are not contained in previous edges in $E_\ell$ implies $\|E_\ell\|\leq O(\eps^{-1}\opt)$. Summation over all $\ell\in \Z$ then yields
\[\|E\|
=\sum_{\ell\in \Z} \|E_\ell\|
=\sum_{\ell=\lfloor\log_\beta(\opt/n^2)\rfloor}^{\lceil \log_\beta \opt\rceil} \|E_\ell\| +O(\opt)
= O(\eps^{-1}\,\opt\log_\beta n).\]

To prove \eqref{eq:1}, consider the containment poset of $E_{\ell}(pq)$. In fact, we represent the poset as a rooted binary tree $T$: The root corresponds to $pq$, and edges $e_1,e_2\in E_\ell(pq)$ are in parent-child relation iff $e_2\subset e_1$, and there is no edge $e'\in E_{\ell}(pq)$ with $e_2\subset e'\subset e_1$.
Each level of $T$ corresponds to interior-disjoint edges contained in $\|pq\|$, so the sum of weight on each level is at most $\|pq\|$.
The total weight of the first $k=\lceil 5\eps^{-1}\rceil$ levels is $O(\eps^{-1}\|pq\|)$.

We claim that the total weight on level $k=\lceil 5\eps^{-1}\rceil$ is at most $\frac12 \|pq\|$. We distinguish between three types of nodes in the subtree of $T$ between levels 0 and $k$: A \emph{branching node} has two children, a \emph{single-child node} has one child, and a \emph{leaf} has no children (in particular all nodes in level $k$ are considered leaves in this subtree). The nodes (leaves) at level $k$ correspond to interior-disjoint edges $e\subset pq$ with $\|e\|\geq \eps\,\|pq\|$ by the definition of $E_\ell(pq)$. Thus there are at most $\lfloor \eps^{-1}\rfloor$ nodes at level $k$, hence there are less than $\lfloor \eps^{-1}\rfloor$ branching nodes. This implies that for any node $e$ on level $k$,
the descending path from the root $pq$ to $e$ contains at least $k-\lfloor\eps^{-1}\rfloor\geq \lceil 4\eps^{-1}\rceil$ single-child nodes.

For the purpose of bounding the total weight at level $k$, we can modify $T$, by incrementally moving all single-child nodes below all branching nodes as follows. While there is an edge $uv$ in $T$, such that $u$ is a branching node, and its parent $v$ is a single-child node, we suppress $u$ and subdivide the two edges of $T$ below $u$ with new nodes $v_1$ and $v_2$. The weight along the edge $uv$ goes down by a factor of at most $(1-\frac{\eps}{4})$ by property (P3); we set the weights in the modified tree such that the same decrease occurs along the edges $uv_1$ and $uv_2$. Then each operation maintains property (P3), and the total weight at level $k$ does not change.
When the while loop terminates, we obtain a full binary tree with a chain attached to each leaf. As we argued above, each chain has length $\lfloor 4/\eps\rfloor$ or more. The full binary tree does not necessarily decrease the weight. Along each chain of $\lfloor 4/\eps\rfloor$ or more single-child nodes, the weight is cumulatively multiplied by a factor of at most $(1-\frac{\eps}{4})^{\lceil 4/\eps\rceil}<\frac12$. Overall, the total weight at level $k=\lceil 5\eps^{-1}\rceil$ is at most $\frac12 \|pq\|$, as claimed.

By induction, for every integer $j\geq 0$, the total weight at level $jk=j\lceil 5\eps^{-1}\rceil$ is at most $\|pq\|/2^j$. Consequently, the total weight of
a \emph{block} of $k$ consecutive levels $\{jk+1,\ldots ,(j+1)k \}$ is at most $k \|pq\|/2^j$.  Overall, $\|E_\ell(pq)\| = \sum_{j\geq 0} k \|pq\|/2^j =O(k\|pq\|) = O(\eps^{-1}\,\|pq\|)$, which completes the proof of \eqref{eq:1}.
\end{proof}

We can summarize the discussion above in the following theorem.

\begin{theorem}\label{thm:1D-bounds}
For every $\eps>0$, the competitive ratio of any online algorithm for $(1+\eps)$-spanners for a sequence of points on a line is $\Omega(\eps^{-1}\log n / \log \eps^{-1})$. Moreover, there is an online algorithm that maintains a $(1+\eps)$-spanner with competitive ratio
$O(\eps^{-1}\log n/ \log \eps^{-1})$.
\end{theorem}

\section{Upper Bounds for Spanners in $\R^d$ under the $L_2$ Norm}
\label{sec:L2-upper}

We turn to online $(1+\eps)$-spanners in Euclidean $d$-space for $d\geq 2$. The dynamic algorithm  {\sc DefSpanner} by Gao et al.~\cite{gao2006deformable}, based on hierarchical clustering, achieves $O(\eps^{-(d+1)}\log n)$ competitive ratio in the online model.
In Section~\ref{sec:L2-upper-woSteiner}, we recover the same bound with a new analysis, where we refine the hierarchical clustering with a partition of the edges into buckets of similar directions, locations, and weights.
In Section~\ref{sec:L2-upper-withSteiner}, we extend the new analysis to show that the competitive ratio improves to $O(\eps^{(1-d)/2}\log n)$ if we are allowed to use Steiner points. Our spanner algorithm replaces each bucket of ``similar'' edges with a Steiner network using grids and shallow-light trees, for up to $O(\eps^{(1-d)/2})$ directions.

\smallskip\noindent\textbf{Preliminaries.}
\emph{Well-separated pair-decomposition} (for short, WSPD) of a finite point set $S$ in a metric space is a classical tool for constructing $(1+\eps)$-spanners~\cite{CallahanK95,GudmundssonK18,narasimhan2007geometric,Smid18}. It is a collection of pairs $\{(A_i,B_i): i\in I\}$ such that for all $i\in I$, we have $A_i,B_i\subset S$ and $\max\{\text{diam}(A_i),\text{diam}(B_i)\}\leq \eps\,\text{dist}(A_i,B_i)$; and for every point pair $\{s,t\}\subset \binom{S}{2}$, there is a pair $(A_i,B_i)$ such that $A_i$ and $B_i$ each contains precisely one of $s$ and $t$. It was shown by Callahan and Kosaraju~\cite{CallahanK93} that if a graph $G=(S,E)$ contains an edge between arbitrary points in $A_i$ and $B_i$, for all $i\in I$, then $G$ is an $(1+O(\eps))$-spanner for $S$; see also \cite[Ch.~9]{narasimhan2007geometric}.

Dynamic spanners (including the fully dynamic algorithm by Roditty~\cite{Roditty12} and {\sc DefSpanner} by Gao et al.~\cite{gao2006deformable}) rely on WSPDs and hierarchical clustering. In $\R^d$, hierarchical clustering can be obtained by classical recursive space partitions such as \emph{quadtrees}~\cite[Ch.~14]{BergCKO08}. Dynamic quadtrees and their variants have been studied extensively, due to their broad range of applications; see~\cite[Ch.~2]{Sariel-Book}. 
In general, dynamic quadtrees can handle both point insertion and deletion operations. However, in the context of an online algorithm, where the points are only inserted, note that no cell of the quadtree is ever deleted.
We analyse the competitive ratio of the dynamic incremental algorithm by Fischer and Har-Peled~\cite{FischerH05} that maintains an $(1+\eps)$-spanner for $n$ points in Euclidean $d$-space in expected $O([\log n+\log \eps^{-1}]\,\eps^{-d}\log n)$ update time. However, they have not analyzed the ratio between the \emph{weight} of the resulting $(1+\eps)$-spanner and the minimum weight of an $(1+\eps)$-spanner.

\subsection{Online Algorithm without Steiner Points}
\label{sec:L2-upper-woSteiner}

\subparagraph{Online Algorithm.}
We briefly review the algorithm in~\cite{FischerH05} and then analyze the weight. The input is a sequence of points $(s_1,s_2,\ldots )$ in $\R^d$; the set of the first $n$ points is denoted by $S_n=\{s_i: 1\leq i\leq n\}$. For every $n$, we dynamically maintain a quadtree $\mathcal{T}_n$ for $S_n$. Every node of $\mathcal{T}_n$ corresponds to a cube.
The root of $\mathcal{T}_n$, at level $0$, corresponds to a cube $Q_0$ of side length $a_0=\Theta(\diam(S_n))$. At every level $\ell\geq 0$, there are at most $2^{d\ell}$ interior-disjoint cubes, each of side length $a_0/2^\ell$. A cube $Q\in \mathcal{T}_n$ is \emph{nonempty} if $Q\cap S_n\neq \emptyset$. For every nonempty cube $Q$, we select an arbitrary representative $s(Q)\in Q\cap S_n$. At each level $\ell$, let $E_\ell$ be the set of all edges $s(Q_1)s(Q_2)$ for pairs of cubes $\{Q_1,Q_2\}$ on level $\ell$ such that $\frac{c_1 a_0}{\eps\,2^\ell}\leq \|s(Q_1) s(Q_2)\|\leq \frac{c_2 a_0}{\eps\,2^\ell}$ for some constants $0<c_1<c_2$ that depend on $d$; see Fig.~\ref{fig:quadtree}(left). The algorithm maintains the spanner $G=(S_n,E)$ where $E=\bigcup_{\ell\geq 0} E_\ell$.
A classical argument by Callahan and Kosaraju~\cite{CallahanK93} (see also~\cite{GudmundssonK18,narasimhan2007geometric,Smid18}) shows that $G$ is a $(1+\eps)$-spanner for $S_n$.

\begin{figure}[htbp]
	\centering
 \includegraphics[width=0.8\textwidth]{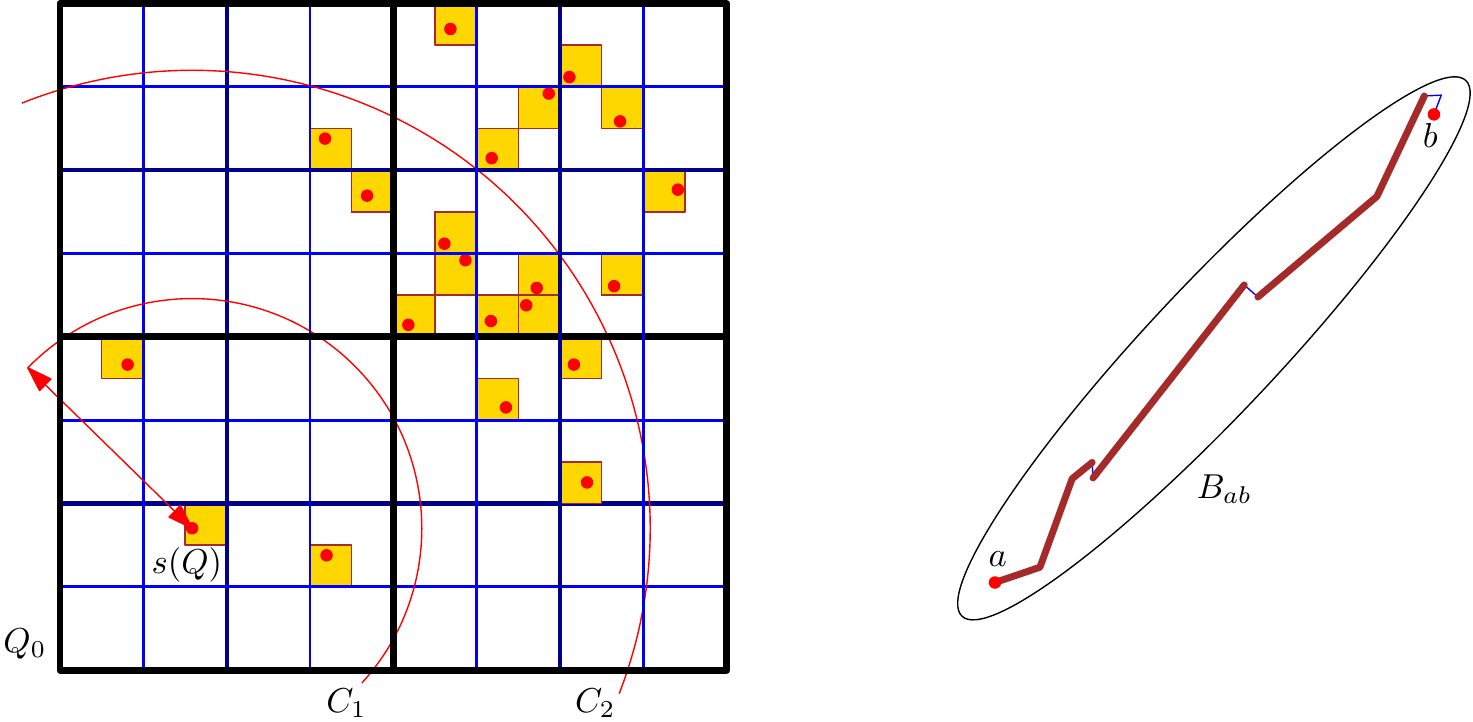}
\caption{Left: Nonempty squares at level $\ell=4$ of a quadtree, each with a representative (red dots). Point $s(Q)$ is connected to all other representatives in the annulus between the concentric circles $C_1$ and $C_2$ of radii $c_1/(\eps\,2^\ell)$ and $c_2/(\eps\,2^\ell)$.
Right: Ellipse $B_{ab}$ with foci $a$ and $b$, an $ab$-path of weight $(1+\eps)\|ab\|$. The bold edges make an angle at most $\eps^{1/2}$ with $ab$.}\label{fig:quadtree}
\end{figure}


\begin{theorem}\label{thm:L2woSteiner}
For every constant $d\geq 2$, parameter $\eps>0$, and a sequence of $n\in \N$ points in Euclidean $d$-space, the competitive ratio of the online algorithm above is in $O(\eps^{-(d+1)}\log n)$.
\end{theorem}
\begin{proof}
For the set $S_n$ of the first $n$ points of a sequence in $\R^d$, let $G=(S_n,E)$ be the $(1+\eps)$-spanner produced by the online algorithm, and let $G^*=(S_n,E^*)$ be an $(1+\eps)$-spanner of minimum weight. We show that $\|G\|/ \|G^*\|=O(\eps^{-(d+1)}\log n)$.

\smallskip\noindent\textbf{Short edges.}
Note that the weight of every edge in $E_\ell\subset E$ at level $\ell$ is $\Theta(\eps^{-1}\diam(S_n)/2^\ell)$, since it connects representatives at $\Theta(\eps^{-1}\diam(S_n)/2^\ell)$ distance apart.
In particular, an edge at any level $\ell\geq 2\log n$ has weight at most $O(\eps^{-1}\diam(S_n)/n^2)$; and the total weight of these edges is $O(\eps^{-1}\diam(S_n))\leq O(\eps^{-1}\opt)$. It remains to bound the weight of the edges on levels $\ell=1,\ldots,\lfloor 2\log n\rfloor$. We consider each level separately.

\smallskip\noindent\textbf{Ellipsoids and directions.}
For every edge $ab\in E$, let $B_{ab}$ denote the ellipsoid with foci $a$ and $b$, and great axis of length $(1+\eps)\|ab\|$. Note that every $ab$-path of weight at most $(1+\eps)\|ab\|$ lies in $B_{ab}$.
The set of directions of line segments in $\R^d$ is represented by a hemisphere of $\sphere^{d-1}$. The distance between two directions is measured by angles in the range $[0,\pi)$. Recently, Bhore and T\'oth~\cite[Lemma~3]{BT-oess-21} proved that every $ab$-path of weight at most $(1+\eps)\|ab\|$ contains edges of total weight at least $\frac12 \|ab\|$ that make an angle at most $\eps^{1/2}$ with $ab$ (i.e., they are near-parallel to $ab$); see Fig.~\ref{fig:quadtree}(right).

Since $G^*$ is a $(1+\eps)$-spanner for $S_n$, it contains an $ab$-path of weight at most $(1+\eps)\|ab\|$ for every $ab\in E$. This path lies in the ellipsoid $B_{ab}$, and contains edges of $G^*$ of weight at least $\frac12\|ab\|$ and with direction with at most $\eps^{1/2}$ from $ab$. We next define suitable \emph{disjoint} sets of ellipsoids, in order to establish a lower bound on $\|G^*\|$.

\smallskip\noindent\textbf{Edge partition by directions.}
First, we partition the edge set $E_\ell$ into subsets based on the \emph{directions} of the edges. We use standard volume argument to construct a \emph{homogeneous} set of directions. Let $H\subset \sphere^{d-1}$ be the hemisphere of unit vectors in $\R^d$, then the direction vector of a line segment $ab$, denoted $\dir(ab)$, is a unique point in $H$. Consider a maximal packing of $H$ with (spherical) balls of radius $\frac{1}{8}\eps^{1/2}$. Since the spherical volume of $H$ is $\Theta(1)$ and the volume of each ball is $\Theta(\eps^{(d-1)/2})$, the number of balls is  $K=\Theta(\eps^{(1-d)/2})$.

By doubling the radii of the spherical balls to  $\frac{1}{4}\eps^{1/2}$, we obtain a covering of $H$ with a set of balls $\mathcal{D}=\{D_i: i=1,\ldots, K\}$. For each spherical ball $D_i\in \mathcal{D}$, denote by $2D_i$ the concentric ball of radius $\frac{1}{2}\eps^{1/2}$. By standard packing argument, the ball $2D_i$ intersects only $O(1)$ balls in $\mathcal{D}$ (where $d=O(1)$). We can now define a partition $E_\ell=\bigcup_{i=1}^K E_{\ell,i}$ as follows: let an $ab\in E_{\ell}$ be in $E_{\ell,i}$ if $i$ is the smallest index such that $\dir(ab)\in D_i$.
Now for every $i=1,\ldots , K$, let $E^*_i$ be the set of edges $e^*\in E^*$
such that $\dir(e^*)\in 2D_i$. By construction, every edge $e^*\in E^*$ lies in $O(1)$
sets $E^*_i$; consequently $\sum_{i=1}^K \|E_i^*\| =\Theta(\|G^*\|)$.
Furthermore, for every edge $ab\in E_{\ell,i}$, all edges in $E^*$ that make an angle at most $\eps^{1/2}$ with $ab$ are in $E^*_i$.

\smallskip\noindent\textbf{Disjoint ellipsoids.}
For every $i=1,\ldots ,K$, let $\mathcal{B}_{\ell,i}$ be the set of ellipsoids $B_{ab}$ with $ab\in E_{\ell,i}$. We show that $\mathcal{B}_{\ell,i}$ contains a subset $\mathcal{B}'_{\ell,i}$ of disjoint ellipsoids such that $|\mathcal{B}'_{\ell,i}|\geq \Omega(\eps^{d+1}|\mathcal{B}_{\ell,i}|)$.

We claim that every ellipsoid in $\mathcal{B}_{\ell,i}$ intersects $O(\eps^{-(d+1)})$ other ellipsoids in $\mathcal{B}_{\ell,i}$. We make use of a volume argument. Let $M_\ell=\max\{\|e\|: e\in E_{\ell}\}$; and note that the side length of every cube at level $\ell$ of the quadtree is $\Theta(\eps\,M_\ell)$.

For every ellipsoid $B_{ab}\in \mathcal{B}_{\ell,i}$, the great axis has length $(1+\eps)\|ab\|$, and the $d-1$ minor axes each have length $\sqrt{(1+\eps)^2-1^2}\|ab\|<2 \eps^{1/2}\|ab\|$, where $\|ab\|\leq M_\ell$. Hence $B_{ab}$ is contained in a cylinder $C_{ab}$ of height $(1+\eps)M_\ell$ whose base is a $(d-1)$-dimensional ball of diameter $2\eps^{1/2}M_\ell$. Any other ellipsoid in $\mathcal{B}_{\ell,i}$ with great axis parallel to $ab$ is contained in a translate of $C_{ab}$. If we rotate $B_{ab}$ about its center by an angle at most $\eps^{1/2}$, then its orthogonal projection to the original great axis decreases, and the maximum distance from the original great axis increases by at most $\|ab\|\frac{1+\eps}{2}\sin\eps^{1/2}< M_\ell\eps^{1/2}$. Consequently, every ellipsoid in $\mathcal{B}_{\ell,i}$ is contained in a translated copy of $2C_{ab}$. Hence, every ellipsoid in $\mathcal{B}_{\ell,i}$ that intersects $B_{ab}$ is contained in $3C_{ab}$. Every cube  at level $\ell$ of the quadtree that intersects $3C_{ab}$ is contained in the Minkowski sum of $3C_{ab}$ and such a cube, which is in turn contained in $4C_{ab}$. Note that the volume of the cylinder $4C_{ab}$ is $O(\eps^{(d-1)/2}M^d_\ell)$; while the volume of a cube at level $\ell$ of the quadtree is $\Theta(\eps^d\,M^d_\ell)$. Therefore $4C_{ab}$ contains  $O(\eps^{(d-1)/2}/\eps^d)=O(\eps^{-(d+1)/2})$ such cubes. Recall that the algorithm maintains one representative from each cube, and the edges $ab\in E_{\ell,i}$ are pairs of representative. Thus $O(\eps^{-(d+1)/2})$ representatives in $4C_{ab}$ can form $O(\eps^{-(d+1)})$ pairs (i.e., edges, hence ellipsoids).

This completes the proof of the claim that every ellipsoid in $\mathcal{B}_{\ell,i}$ intersects $O(\eps^{-(d+1)})$ other ellipsoids in $\mathcal{B}_{\ell,i}$. Hence the \emph{intersection graph} of $\mathcal{B}_{\ell,i}$ is $O(\eps^{-(d+1)})$-degenerate;
and has an independent set $\mathcal{B}'_{\ell,i}$ of size  $|\mathcal{B}'_{\ell,i}|\geq \Omega(\eps^{d+1}|\mathcal{B}_{\ell,i}|)
=\Omega(\eps^{d+1} |E_{\ell,i}|)$.

\smallskip\noindent\textbf{Weight analysis.}
As noted above, all edges in $E_\ell$ have length $\Theta(M_\ell)$.
For every $i=1,\ldots , K$ and for every ellipsoid $B_{ab}\in \mathcal{B}_{\ell,i}$,
we have $\|E^*_i\cap B_{ab}\|\geq \frac12 \|ab\|\, \Omega(M_\ell)$.
Summing over a set of disjoint ellipsoids, we obtain
\begin{align*}
\|E^*_i\|
&\geq \sum_{B_{ab}\in \mathcal{B}'_{\ell,i}} \|E^*_i\cap B_{ab}\|
\geq \sum_{B_{ab}\in \mathcal{B}'_{\ell,i}} \frac12 \|ab\| \\
&\geq |\mathcal{B}'_{\ell,i}|\cdot \frac12\, \min\{\|ab\|: ab\in E_{\ell,i}\} \\
&\geq \eps^{-(d+1)} |E_{\ell,i}|\cdot \Omega(M_\ell)
=\Omega(\eps^{-(d+1)} \|E_{\ell,i}\|).
\end{align*}
Summation over all directions $i=1,\ldots , K$ yields
\[\|G^*\| =\Theta\left(\sum_{i=1}^{K}\|E^*_i\|\right)
\geq \Omega\left(\sum_{i=1}^K \eps^{-(d+1)} \|E_{\ell,i}\|\right)
= \Omega(\eps^{-(d+1)} \|E_\ell \|).
\]
Finally, summation over all $\ell\geq 1$ yields
\[
\|E\|
=\sum_{\ell\geq 1}\|E_{\ell}\|
\leq \sum_{\ell=1}^{\lfloor 2\log n\rfloor} \|E_\ell\|
    + \sum_{\ell>\lfloor 2\log n\rfloor} \|E_\ell\|
\leq \eps^{-(d+1)}\|G^*\|\log n + \eps^{-1}\|G^*\|,
\]
as required.
\end{proof}

\subsection{Online Algorithm with Steiner Points}
\label{sec:L2-upper-withSteiner}

When Steiner points are allowed, we can substantially improve the competitive ratio in terms of $\eps$. We describe an algorithm with competitive ratio $O(\eps^{(1-d)/2}\log n)$. As a counterpart, we show in Section~\ref{sec:L2-lower-withSteiner} that the dependence on $n$ is unavoidable in dimensions $d\geq 2$; it remains an open problem whether the dependence on $\eps$ is necessary.


\begin{theorem}\label{thm:L2withSteiner2D}
For every $\eps>0$, an online algorithm can maintain, for a sequence of $n\in \N$ points in the plane, a Euclidean Steiner $(1+\eps)$-spanner of weight $O(\eps^{-1/2}\log n)\cdot \opt$.
\end{theorem}
\begin{proof}
Our online algorithm has two \emph{stages}: $A_1$ and $A_2$. Algorithm $A_1$ is the same as in Section~\ref{sec:L2-upper-woSteiner}, it maintains a quadtree $\mathcal{T}_n$ for the point set $S_n$, and a ``primary'' $(1+\eps)$-spanner $G_1$ \emph{without} Steiner points. Algorithm $A_2$ maintains a Steiner $(1+3\eps)$-spanner $G_2$ as follows: for each edge $ab$ in $G_1$, it creates an $ab$-path of length $(1+\eps)\|ab\|$ using Steiner points in $G_2$.
Importantly, algorithm $A_2$ can bundle together ``similar'' edges of $G_1$, and handle them together using shallow-light trees~\cite{Solomon15}.

In particular, we partition the space of all
possible edges of $G_1$ into \emph{buckets} (edges with similar directions, locations, and weights). For each bucket $U$, when algorithm $A_1$ inserts the first edge $ab\in U$ into $G_1$, then algorithm $A_2$ creates a ``backbone'' Steiner tree $T=T(U)$ of weight $O(\|ab\|)$, which contains an $ab$-path of length at most $(1+\eps)\|ab\|$.
For any subsequent edge $a'b'\in U$, is suffices to add paths from $a'$ and $b'$ to $T$, of weight $O(\eps\,\|ab\|)$, to obtain $a'b'$-path of length at most  $(1+\eps)\|a'b'\|$.
Overall, between any two points $s_i,s_j\in S$, the primary spanner contains a path of weight at most $(1+\eps)\|s_is_j\|$, and $G_2$ contains an Steiner path of weight at most $(1+\eps)^2\|s_is_j\|<(1+3\eps)\|s_is_j\|$, as claimed.

It remains to define the buckets $U$, the backbone $T(U)$ for the first edge in $U$, and the ``connectors'' added for each subsequent edge in $U$.
We first describe the algorithm in the plane, where we establish a competitive ratio $O(\eps^{-1/2}\log n)$, and then generalize the construction to higher dimensions.
%


\smallskip\noindent\textbf{Buckets.} We define buckets for all potential edges in the primary spanner $G_1$. We analyze a single level $\ell$ of the quadtree $\mathcal{T}$. Without loss of generality, assume that the side length of all quadtree cubes in level $\ell$ have unit length, hence the weight of every edge in $E_\ell$ is $\Theta(\eps^{-1})$.

In Section~\ref{sec:L2-upper-withSteiner}, we have covered the set $H\subset \sphere^{1}$ of directions with a set $\mathcal{D}=\{D_i:i=1,\ldots K\}$ of balls of diameter $\eps^{1/2}$. For each ball in $\mathcal{D}$, we define a set of buckets. Let $D\in \mathcal{D}$, and let $L$ be a line such that $\dir(L)$ corresponds to the center of $D$; refer to Fig.~\ref{fig:stripes}(left). Partition the plane into parallel strips of width $\frac12\, \eps^{1/2}$ by a set of lines parallel to $L$; and partition each strip further into rectangles of height $2\eps^{-1}$. By scaling up the rectangles by a factor of 2, we obtain a covering of the square $Q_0$ with a set $\mathcal{R}$ of $4\eps^{-1}\times \eps^{1/2}$ rectangles such that each point is covered by $O(1)$ rectangles in $\mathcal{R}$.

For each rectangle $R\in \mathcal{R}$, we create a bucket $U$ comprising all edges $ab\in E_\ell$ such that $ab\subset R$ and $\dir(ab)\in D$ (hence $\angle(\dir(ab),\dir(L))\leq \eps^{1/2}$). Note that every edge $ab\in E_{\ell}$ lies in at least one and at most $O(1)$ buckets.

\begin{figure}[htbp]
	\centering
 \includegraphics[width=0.9\textwidth]{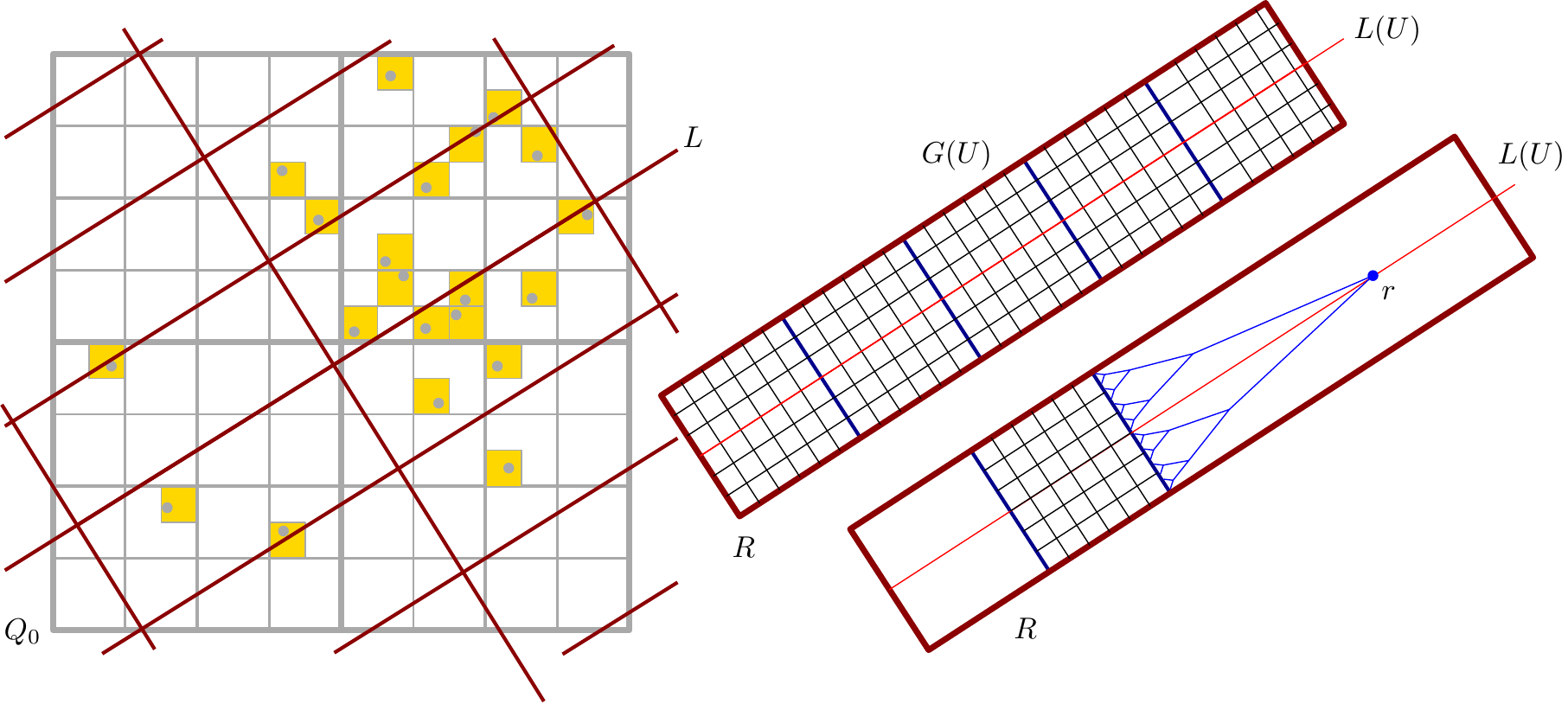}
\caption{Left: The overlay the the quadtree with a partition of $\R^2$ into  $\frac12\eps^{-1/2}\times 2\eps^{-1}$ rectangles aligned with $L$.
Top-Right: A rectangle $R\in \mathcal{R}$, the median $L(U)$, the grid $G(U)$, and the partition of $R$ into $\eps^{-1/2}\times \eps^{-1/2}$ squares.
Bottom-Right: A shallow-light tree between a side of an $\frac12\eps^{-1/2}\times \frac12\eps^{-1/2}$ square and a source $r\in L(U)$.}\label{fig:stripes}
\end{figure}

\smallskip\noindent\textbf{Backbones and Connectors.}
Let $U$ be a bucket defined above for a rectangle $R\in \mathcal{R}$. Let $L(U)$ denote the median of the rectangle $R$ parallel to $L$. When the primary algorithm $A_1$ inserts the first edge $ab\in U$ into $G_1$, then Algorithm $A_2$ constructs a unit grid graph $G(U)$, formed by a subdivision of $R$ into unit squares; see Fig.~\ref{fig:stripes}(top-right). Since $R$ is a $4\eps^{-1}\times \eps^{-1/2}$ rectangle, $\|G(U)\|= O(\eps^{-3/2})$.
Furthermore, we partition $R$ into $4\eps^{-1/2}$ squares of side length
$\eps^{-1/2}$. For each such square, we insert two shallow-light trees~\cite{Solomon15} between the two sides of the square orthogonal to $L$ and two points in $L(U)$ at distance $\eps^{-1}$ from the square on either side; Fig.~\ref{fig:stripes}(bottom-right).
The weight of each shallow-light tree is $O(\eps^{-1})$~\cite{Solomon15}, and so the combined weight of $O(\eps^{-1/2})$ shallow-light trees is $O(\eps^{-3/2})$.
The grid $G(U)$ together with the shallow-light trees forms the \emph{backbone} for the bucket $U$ in $G_2$.

We add \emph{connector} edges between $a$ (resp., $b$) and the four corners of unit square of the grid $G(U)$ that contains it.
For any subsequent edge $a'b'\in U$ that algorithm $A_1$ inserts into $G_1$,
the backbone does not change, we only add connectors between $a'$ (resp., $b'$)
and the four corners of the unit square in $G(U)$ that contains it. The weight of
the four connectors is $O(1)$ per point. Since $\text{area}(R)=\Theta(\eps^{-3/2})$,
then $R$ intersects at most $O(\eps^{-3/2})$ unit squares of the quadtree at level $\ell$, and so the total weight of all connectors is $O(\eps^{-3/2})$, as well.

\smallskip\noindent\textbf{Stretch analysis.} Suppose algorithm $A_1$ inserts an edge $cd$ into $G_1$. As noted above, $cd$ lies in $\Theta(1)$ buckets; refer to Fig.~\ref{fig:stretch}. Suppose bucket $U$ contains $cd$; and in the partition of the rectangle $R=R(U)$, the endpoint $c$ ($d$) lies squares $R_c$ ($R_d$) of side length $\eps^{-1/2}$, associated with shallow-light trees rooted at $r_c$ ($r_d$). Then $G_2$ contains a $cd$-path comprised of: (i) connectors from $c$ and $d$, resp., to the closest point in the grid $G(U)$; (ii) paths in $G(U)$ from the connectors to the boundary of squares $R_c$ and $R_d$, (iii)  paths along the shallow-light trees to the roots $r_c,r_d\in L(U)$, and (iv) the line segment $r_c r_d$ in $G(U)$. The weight of each connector in (i) is at most $2\sqrt{2}$, which is bounded by $O(\eps)\|cd\|$ since  $\|cd\|=\Theta(\eps^{-1})$. The edges in (ii) and (iv) are parallel to $L$, hence they make an angle less than $\eps^{1/2}$ with $cd$. Finally, consider the two subpaths in part (iii) in shallow-light trees: The line segment between the two endpoints of each such subpath makes an angle less than $\eps^{1/2}$ with $L$, hence less than $2\eps^{1/2}$ with $cd$; and the weight of a root-to-leaf path in a shallow-light tree is a $(1+\eps)$-approximation of the straight-line segment between its endpoints. Overall, the total weight of the $cd$-path described above is $(1+O(\eps))\|cd\|$, as required.

\begin{figure}[htbp]
	\centering
 \includegraphics[width=0.98\textwidth]{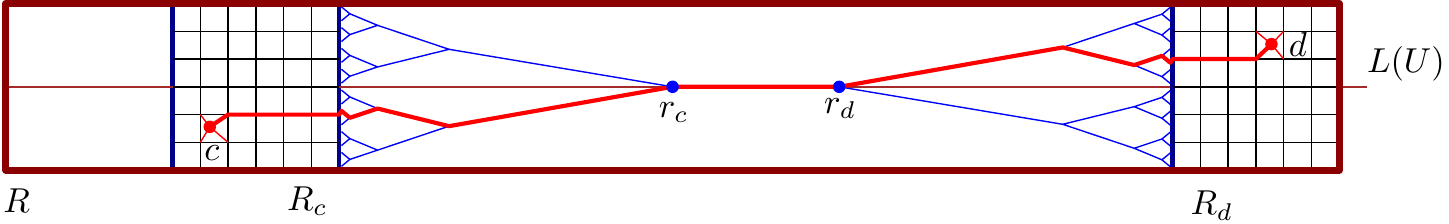}
\caption{A $cd$-path in the Steiner spanner $G_2$.}\label{fig:stretch}
\end{figure}

For every point pair $a,b\in S_n$, the primary graph $G_1$ contains an $ab$-path $P=(p_0,\ldots , p_m)$ of length $\|P\|\leq (1+\eps)\|ab\|$, since $G_1$ is a $(1+\eps)$-spanner. We have shown that for every edge $p_{i-1}p_i$ of $G_1$, the Steiner spanner $G_2$ contains a $p_{i-1}p_i$-path of weight $(1+O(\eps))\|p_{i-1}p_i\|$. The concatenation of these paths yields an $ab$-path in $G_2$, of weight $\sum_{i=1}^m (1+O(\eps))\|p_{i-1}p_i\| =(1+O(\eps))\|P\| = (1+O(\eps))(1+\eps)\|ab\| = (1+O(\eps))\|ab\|$.

\subparagraph{Competitive Analysis.}
Denote by $E_\ell$ the set of edges of $G_2$ added at level $\ell=1,\ldots  ,2\log n$, and let $b_\ell$ be the number of nonempty buckets at level $\ell$. We have seen that for each nonempty bucket at level $\ell$, $E_\ell$ contains a subgraph of weight $O(\eps^{-3/2} \diam(S_n)/2^\ell)$; hence $\|E_\ell\|\leq O(b_\ell\cdot \eps^{-3/2} \diam(S_n)/2^\ell)$.

Let $G^*=(S_n,E^*)$ the a Euclidean Steiner $(1+\eps)$-spanner for $S_n$ of minimum weight $\opt$. Consider a nonempty bucket $U$ associated with a line $L$ and a rectangle $R(U)$. Since $U$ is nonempty, there is an edge $ab\in U$ in $G_1$. Recall that $ab\in R$ and $\angle(\dir(ab),\dir(L))\leq \eps^{1/2}$. Since $G^*$ is a $(1+\eps)$-spanner, it contains an $ab$-path $P_{ab}$ of weight at most $(1+\eps)\|ab\|$. As noted in Section~\ref{sec:L2-upper-woSteiner}, $P_{ab}$ lies in the ellipse $B_{ab}$, and contains edges of weight at least $\frac12\|ab\|$ that make an angle at most $\eps^{1/2}$ with $ab$. All points in the ellipse $B_{ab}$ are at distance less than $\eps^{1/2}$ from the the line segment $ab$. The segment $ab$ lies in the $4\eps^{-1}\times \eps^{1/2}$ rectangle $R(U)$. Thus we have $P_{ab}\subset B_{ab}\subset 2R(U)$, and so $2R(U)$ contains edges of $G^*$ of weight $\frac12\|ab\|=\Omega(\eps^{-1}\diam(S_n)/2^\ell)$ whose directions are within $2\eps^{1/2}$ from $L$; denote by $E^*(U)\subset E^*$ the set of these edges. By construction, each edge $e^*$ of $G^*$ lies in $E^*(U)$ for only $O(1)$ buckets. Indeed, there  are $O(1)$ lines $L'$ with $\angle(\dir(L),\dir(L))\leq 2\eps^{1/2})$, and for each such direction $L'$, every point in $\R^2$ lies in $O(1)$ rectangles $2R(U')$ aligned with $L'$. We conclude that $\opt = \|G^*\| = \Omega(b_\ell\cdot \eps^{-1}\diam(S_n)/2^\ell)$.
This implies $\|E_\ell\|/\opt\leq O(\eps^{-1/2})$ for $\ell=1,\ldots ,2\log n$.
Summation over all levels yields
\[\frac{\alg}{\opt}
= \frac{\sum_{\ell=1}^\infty \|E_\ell\|}{\opt}
\leq \sum_{\ell=1}^{2\log n}O(\eps^{-1/2}) +O(1) =O(\eps^{-1/2}\log n),\]
as claimed.
\end{proof}

\subparagraph{Generalization to $\R^d$.} Our algorithm and its analysis generalize to Euclidean $d$-space.

\begin{theorem}\label{thm:L2withSteiner}
For every $\eps>0$, an online algorithm can maintain, for a sequence of $n\in \N$ points in $\R^d$, a Euclidean Steiner $(1+\eps)$-spanner of weight $O(\eps^{(1-d)/2}\log n)\cdot \opt$.
\end{theorem}
\begin{proof}[Proof sketch.]
The proof is analogous to that of Theorem~\ref{thm:L2withSteiner2D}, we highlight only the differences in the algorithm and its analysis. The buttleneck of the competitive analysis is the size of the unit grids $G(U)$ which is $\Theta(\eps^{-(d+1)/2})$ in $\R^d$, which is contrasted with a path of weight $\Omega(\eps^{-1})$ in $\opt$.

Similarly to Section~\ref{sec:L2-upper-woSteiner}, we choose a homogeneous set $D$ of $\Theta(\eps^{(1-d)/2})$ directions (i.e., any direction is within
angle $\eps^{1/2}$ from a direction in $D$, and the angle between any two directions in $D$ is at least $\frac12\eps^{1/2}$). For each direction $L\in D$, we construct a tiling of $\R^d$ with congruent hyper-rectangles aligned with $L$ of dimensions $\eps^{-1}\times \eps^{-1/2}\times \ldots \eps^{-1/2}$; and a covering of $Q_0$ after scaling up the hyperrectangles by a factor of 2. We associate a bucket to each hyperrectangle $R$ in the covering: an edge $ab$ of $G_1$ is in bucket $U$ if $ab\subset R$ and $\angle (\dir(ab),\dir(L))\leq \eps^{1/2}$. The construction ensures that every edge $ab\in E_{\ell}$ is in at least one bucket, and at most $O(1)$ buckets.

For each nonempty bucket $U$, the grid $G(U)$, shallow-light trees~\cite{Solomon15}, and the connectors are analogous to the planar construction. However, the weight of the unit grid is $\|G(U)\|=\Theta(\eps^{-(d+1)/2})$ in $\R^d$. The stretch analysis carries over to higher dimensions. The lower bound for $\opt$ is the same as in the plane: for each nonempty bucket at level $\ell$, the rectangle $2R(U)$ contains edges of $G_2$ of weight $\Omega(\eps^{-1}\diam(S_n)/2^\ell)$ with direction within $2\eps^{1/2}$ from $L$.
This yields an upper bound $\|E_\ell\|/\opt\leq O(\eps^{(1-d)/2})$ for levels $\ell=1,\ldots ,2\log n$; and $\alg/\opt\leq O(\eps^{(1-d)/2}\log n)$ overall.
\end{proof}

\section{Lower Bound with Steiner Points}
\label{sec:L2-lower-withSteiner}

Recall that when Steiner points are allowed, the algorithm may subdivide existing edges with Steiner points. It follows that the in one-dimension, an online algorithm can maintain a Hamiltonian path on $S_n$, which is the minimum $(1+\eps)$ spanner for all $\eps\geq 0$. This property carries over to Euclidean Steiner $1$-spanners (i.e., the case  $\eps=0$), where we need to maintain the complete straight-line graph on $n$ points. However, we show that for $\eps>0$ in dimensions $d\geq 2$, the competitive ratio of an online algorithm with Steiner points must depend on $n$.

\begin{theorem}\label{thm:w-SP-LB}
For every $\eps>0$, the competitive ratio of any online algorithm that maintains a Euclidean Steiner $(1+\eps)$-spanner for a sequence of $n$ points in $\R^d$ is $\Omega(f(n))$ for some function $f(n)$ such that $\lim_{n\rightarrow \infty}f(n)=\infty$.
\end{theorem}
\begin{proof}
We describe and analyze an adversarial strategy for placing points in the plane in \emph{stages}. In stage~1, the adversary places two points at $s=(0,0)$ and $t=(1,0)$, both on the $x$-axis. In subsequent stages, new points are arranged so that the optimum solution remains an $x$-monotone path of length at most $1+\eps$ at all times.

Let us denote by $A_i$ the points placed in stage $i$. At the end of stage $i$, adversary constructs the point set $A_{i+1}$ based on the current $(1+\eps)$-spanner built by the algorithm \alg, and then placed the points in $A_{i+1}$ in an arbitrary order.
The objective is that in each stage, \alg\ has to add new edges of total weight at least $1/2$. Since $\opt\leq 1+\eps$ at all times, and $\alg\geq \frac12(i+1)$ after $i$ stages, the competitive ratio goes to infinity.

We describe stage 2 in more detail; subsequent stages are similar; see Fig.~\ref{fig:stages}. At the end of stage~1, our point set is $A_1=\{s=(0,0), t=(1,0)\}$, the optimal spanner is a single edge of unit weight, and \alg\ has constructed a Euclidean Steiner $(1+\eps)$-spanner $G_1$ for $A_1$. Let $k_1=\left\lceil \|G_1\|\right\rceil$. The adversary considers $2k_1+1$ circular arcs between $s$ ad $t$, each of weight at most $1+\frac{\eps}{2}$. The arcs define $2k_1$ interior-disjoint bounded regions. Let $R_1$ be a region that minimizes the weight $\|G_1\cap R_1\|$, in particular, $\|G_1\cap R_1\|\leq \frac{1}{2k_1}\,\|G_1\|\leq \frac12$. In the interior of $R_1$, let $\gamma_1$ be another circular arc between $s$ and $t$, of weight $\|\gamma_1\|\leq 1+\frac{\eps}{2}$; and let $A_2=\{t_1,\ldots , t_N\}$ be a set of points along $\gamma_1$, labeled in $x$-monotone increasing order with the following properties: (1) For every $i=1,\ldots ,N-1$, the ellipse $B_i$ with foci $t_i$ and $t_{i+1}$, and great axis $(1+\eps)\|t_it_{i+1}\|$ lies entirely in $R_1$; and (2) the weight of the $x$-monotone path $(t_1,t_2,\ldots , t_N)$ is at least 1.

\begin{figure}[htbp]
	\centering
 \includegraphics[width=0.95\textwidth]{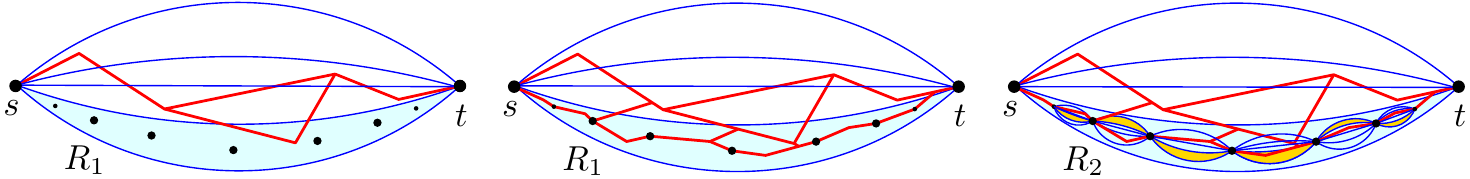}
\caption{Left: For $A_1=\{s,t\}$, $\alg$ constructs a $(1+\eps)$-spanner $G_1$ (red). Five circular arcs define four regions; region $R_1$ satisfies $\|G_1\cap R_1\|\leq \frac14 \|G_1\|$.  In stage~2, the adversary presents points $A_2$ in $R_1$
Middle: The algorithm augments $G_1$ to $G_2$.
Right: Region $R_2$ satisfies $\|G_2\cap R_2\|\leq \frac{1}{2k_2}\|G_2\|$.}\label{fig:stages}
\end{figure}

In stage~2, the adversary presents the points in $A_2$ in an arbitrary order. By the end of stage~2, \alg\ augments $G_1$ to a Euclidean Steiner $(1+\eps)$-spanner $G_2$ for $A_1\cup A_2$. In particular, for every $i=1,\ldots , N-1$, the graph $G_2$ contains a
$t_it_{i+1}$-path of length at most $(1+\eps)\|t_it_{i+1}\|$, which lies in the ellipse $E_i$, hence in the interior of the region $R_1$. The part of the path between the vertical lines passing through $t_i$ and $t_{i+1}$ has weight at least $\|t_it_{i+1}\|$. Since these parts are disjoint, the total weight all $N-1$ paths is lat least $\sum_{i=1}^{N-1}\|t_it_{i+1}\|\geq 1$.  Consequently, $\|G_2\cap R_1\|\geq 1$. Since we had $\|G_1\cap R_1\|\leq \frac12$, \alg\ must have added new edges of weight at least $\frac12$ in stage~2, as claimed.

In phase~$i+1$, in general, let $k_i=\left\lceil \|G_i\|\right\rceil$. Label the points in the current point set $S=\bigcup_{j=1}^i A_j$ by $s_0,\ldots ,s_n$ in $x$-monotone order, and assume that the $x$-monotone path spanned by $S$ has weight $\opt=1+(1-\frac{1}{2^i})\eps$. For all segments $s_js_{j+1}$, we consider $2k_i+1$ $x$-monotone circular arcs such that the total weight of any concatenation of the circular arcs from $s=s_0$ to $t=s_n$ is at most $1+(1-\frac{1}{2^{i+1}})\eps$.
For each segment $s_js_{j+1}$, we choose one of $2k_i$ regions that has a minimum-weight intersection with $G_i$, and let $R_i$ be the union of these regions. Note that $\|G_i\cap R_i\|\leq \frac{1}{2k_i}\,\|G_i\|\leq \frac12$. Let $\gamma_i$ be an $st$-path $\gamma_i$ that connects the points $s_0,\ldots , s_n$ via circular arcs in the region $R_i$, and has weight at most $1+(1-\frac{1}{2^{i+1}})\eps$. Now the adversary can choose a finite point set $A_{i+1}=\{t_1,\ldots , t_N\}$ along $\gamma_i$ with properties (1)--(2) above. This completes the description of the adversarial strategy.

Similarly to stage~2, when \alg\ augments $G_i$ to a Euclidean Steiner $(1+\eps)$-spanner $G_{i+1}$ for $\bigcup_{j=1}^{i+1} A_j$, he must add new edges of weight at least $\frac12$ in the region $R_i$. It follows that the competitive ratio for any online algorithm goes to infinity as $n$ goes to infinity.
\end{proof}


\section{Lower Bounds for Spanners in $\R^d$ under the $L_1$ Norm}\label{sec:lower-bound-L1}

In this section, we study the online $(1+\eps)$-spanners for points in $\mathbb{R}^d$, under the $L_1$ norm. Here the distance between any pair of points is the sum of absolute difference between the coordinates in all dimensions. For instance in $\mathbb{R}^2$, the distance between points $p=(p_1,p_2)$ and $q=(q_1,q_2)$ is the rectilinear distance  between them, i.e., $|p_1 - q_1|+|p_2 - q_2|$.

\begin{figure}[htbp]
    \centering
    \includegraphics[width=\textwidth]{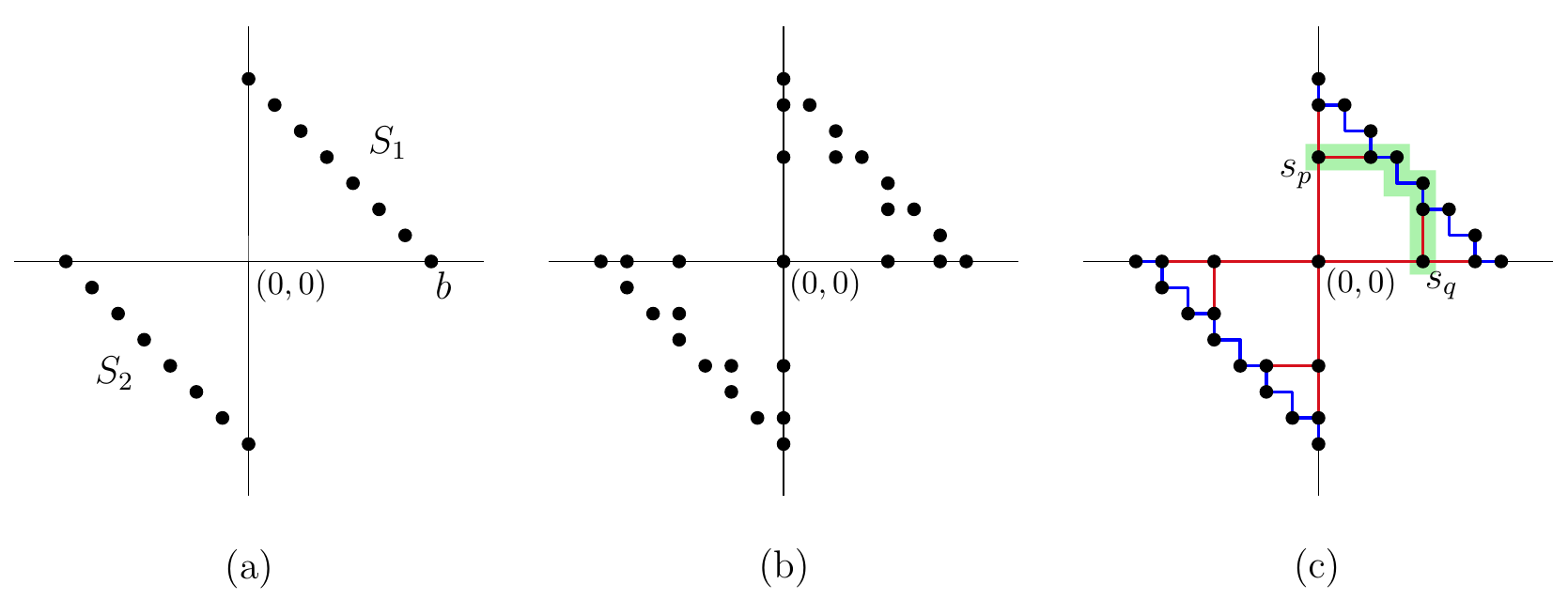}
    \caption{(a) The point sets $S_1$ and $S_2$ in two opposite quadrants (b) The entire point set is $(S_1\cup \widehat{S}_1)\cup (S_2\cup \widehat{S}_2)$; (c) A Manhattan network which is $(1+\eps)$-spanner consisting of two binary trees $T_1$ and $T_2$ and two staircase paths $P_1$ and $P_2$. The green path is an $xy$-monotone path between point pairs $s_p$ and $s_q$.}
    \label{fig:l1-spanner}
\end{figure}

\subparagraph{Construction.} First, we describe the construction for points in $\mathbb{R}^2$. We use the following adversarial strategy to build the construction. 

Let $k=\left\lceil \log \eps^{-1}\right\rceil$.
The adversary places $2^{k}$ collinear points in the positive quadrant, and $2^{k}$ points in the negative quadrant such that the points are placed in a quadrant maintain uniform spacing between consecutive points.
Let $S_1=\{s_i: i=0,\ldots,2^{k}-1\}$, where $s_i=(i,2^k-i)$. Then, $S_2$ is a reflected copy of $S_1$ about the origin. 
See Fig.~\ref{fig:l1-spanner}(a) for an illustration.
Next, the adversary introduces a set of additional points $\widehat{S}_1\cup \widehat{S}_2$ in order to reduce the weight of an optimum solution. Consider the set $S_1$. We consider the binary partition of $\{1,\ldots,2^{k}\}$ into intervals, associated with a binary tree. At level~$0$, the root corresponds to the interval $[1,2^{k}]$.
Then, at each level~$j$ we have the intervals
$[i\cdot 2^{k-j}+1, (i+1)\cdot 2^{k-j}]$, for $i=0,\ldots,2^j$. For each such interval, the adversary considers the bounding box of the points lying in this interval, and adds the lower-left corner into $\widehat{S}_1$. Then, $\widehat{S}_2$ is a reflected copy of $\widehat{S}_1$ about the origin. The entire point set is $(S_1\cup \widehat{S}_1)\cup (S_2\cup \widehat{S}_2)$;
see Fig.~\ref{fig:l1-spanner}(b)


\subparagraph{Competitive Ratio.}
For the point set $(S_1\cup \widehat{S}_1)\cup (S_2\cup \widehat{S}_2)$, we describe a Manhattan network which is a $(1+\eps)$-spanner. The network is comprised of the paths $P_1$ and $P_2$ and the two binary trees $T_1$ and $T_2$ on the point set $(S_1\cup \widehat{S}_1)$ and $(S_2\cup \widehat{S}_2)$, respectively; see Fig.~\ref{fig:l1-spanner}(c). First, we argue that why it is a Manhattan network.
In order to show this, we need to argue that for each pair of points in $(S_1\cup \widehat{S}_1)\cup (S_2\cup \widehat{S}_2)$, there is a rectilinear $xy$-monotone path of optimum length comprised of horizontal and vertical segments.

Consider a pair of points $s_p,s_q \in (S_1\cup \widehat{S}_1)$.
Without loss of generality, assume that $y(s_p)>y(s_q)$. If, $s_q$ is a descendent of $s_p$ in $T_1$, then we are done since there is an unique $xy$-monotone path between them in $T_1$. Otherwise, we take the left path form $s_p$ to the leaf point in $S_1$ and the right path from $s_q$ to the leaf point in $S_1$, and the sub-path from $P_1$ between these leaf points. These paths together from a $xy$-monotone rectilinear path; see Fig.~\ref{fig:l1-spanner}. Moreover, if both $s_p$ and $s_q$ are leaves, then we have a $xy$-monotone path between them which is a subpath of $P_1$.
The same arguments holds for any point pair in $(S_2\cup \widehat{S}_2)$ for which we have the complete binary tree $T_2$ and the path $P_2$. Now, what is left two show that there is an $xy$-monotone path between any pair of points $s_p,s_q$, where $s_p\in (S_1\cup \widehat{S}_1)$ and $s_q\in (S_2\cup \widehat{S}_2)$. For this, we consider the unique $xy$-monotone path from $s_p$ to the root in $T_1$ and $xy$-monotone path from $s_q$ to the root in $T_2$. We take the union of them, which is a $xy$-monotone rectilinear path between $s_p$ and $s_q$.

The weight of the Manhattan network is the weight of the two trees $T_1$ and $T_2$ and the two paths $P_1$ and $P_2$. The distance between the root and $S_1$ is $\min\{|i|+|2^k-i|: i=0,\ldots , 2^k-1\}=2^k$. At each level~$j$, we construct $2^j$ segments. Each segment at level $j$ is of weight $2^{k-j}$. Summation over all levels yields $O(k\,2^k)=O(\eps^{-1}\log \eps^{-1})$.

For the point set $S_1\cup S_2$, every $(1+\eps)$-spanner contains a complete bipartite graphs between $S_1$ and $S_2$. Indeed, the $L_1$-distance between any two points in $S_1$ (resp., $S_2$) is 2 or more, and the $L_1$ distance between any points in $S_1$ and $S_2$ is exactly $\text{dist}(S_1,S_2)=2(2^k-1)$. Hence $(1+\eps)\|s_p s_q\|\leq (1+2^{-k})\cdot 2(2^k-1)< \text{dist}(S_1,S_2)+2$. However, the weight of any $s_p s_q$-path via a third point in $S_1\cup S_2$ would be at least
$\text{dist}(S_1,S_2)+2$.
Therefore, $\alg \geq \Omega(|S_1|\cdot |S_2|\cdot \text{dist}(S_1,S_2)) =\Omega(\eps^{-3})$ for of any online algorithm.
Contrasted with the upper bound $\opt\leq O(\eps^{-1}\log\eps^{-1})$ for $(S_1\cup \widehat{S}_1)\cup (S_2\cup \widehat{S}_2)$, the competitive ratio is $\alg/\opt\geq \Omega(\eps^{-2}/\log \eps^{-1})$. We can conclude the following theorem.

\begin{theorem}
For every $\eps>0$, the competitive ratio of any online algorithm for $(1+\eps)$-spanners in $\R^2$ under the $L_1$ norm is $\Omega(\eps^{-2}/\log \eps^{-1})$.
\end{theorem}

\subparagraph{Construction in $\mathbb{R}^d$.}
For a given $\eps>0$, let $k=\left\lceil\log \eps^{-1}\right \rceil$. First, the adversary introduces point sets $S_1$ and $S_2$ on two opposite faces of a cross polytope: Let $S_1$ be the set of all points with non-negative integer coordinates in the hyperplane $H_1:\sum_{i=1}^d x_i=2^k-2$, and let $S_2$ be the reflected image of $S_1$ in the origin.  Note that $|S_1|=|S_2|=\Theta(2^{k(d-1)})=\Theta(\eps^{1-d})$, and the bounding box of $S_1$ is $Q=[0,2^k-2]^d$. Next, the adversary introduces a set of additional points $\widehat{S}_1\cup \widehat{S}_2$ in order to reduce to weight of an optimum solution. Consider the quadtree $\mathcal{Q}_1$ for $S_1$, which partitions $Q$ into congruent cubes and recurses on all nonempty subcubes. Let $\widehat{S}_1$ be the set of all vertices of the cubes in the quadtree; and let $\widehat{S}_2$ be the reflected image of $\widehat{S}_1$ in the origin. By construction, all points in
$(S_1\cup \widehat{S}_1) \cup (S_2\cup \widehat{S}_2)$ have integer coordinates.

\subparagraph{Competitive Ratio.}
Similarly to the planar construction, every $(1+\eps)$-spanner will contain a complete bipartite graph between $S_1$ and $S_2$. As $\|s_p s_q\|\geq 2^k=\Omega(\eps^{-1})$ for all $s_p\in S_1$ and $s_q\in S_2$, then $\alg\geq \Omega(\eps^{-1}|S_1|\cdot |S_2|)\geq \Omega(\eps^{1-2d})$.

Next, we give a lower bound for $\opt$ for the entire point set $(S_1\cup S_2)\cup (\widehat{S}_1\cup\widehat{S}_2)$. Let $T_1$ is the graph formed by all edges of the quadtree $\mathcal{Q}_1$, and $T_2$ is the reflected image of $T_1$. Since the depth of $\mathcal{Q}_1$ is $k$, each nonleaf node has $3$ or more children, and the edge lengths of the cubes in $\mathcal{Q}_1$ decrease by factors of $2$. Consequently, the weight $\|T_1\|$ is dominated by the weight of level $k$, which contains $\Theta(|S_1|)=\Theta(2^{k(d-1)})=\Theta(\eps^{1-d})$ unit cubes, hence the total weight of their (unit-length) edges is also $\Theta(\eps^{1-d})$. Overall, we have $\|T_1\|=\Theta(\eps^{1-d})$.

We show that $T_1\cup T_2$ is a Manhattan network for $(S_1\cup S_2)\cup (\widehat{S}_1\cup\widehat{S}_2)$.
Consider first $T_1$. At each level $\ell=0,\ldots , k$ the quadtree $\mathcal{Q}_1$, contains all cubes that intersect the hyperplane  $\sum_{i=1}^d x_i=2^k-2$. These cubes have the same size. For any two vertices of two cubes on level $\ell$ of $\mathcal{Q}_1$, say $s_a\in Q_a$ and $s_b\in Q_b$, let $s'_a$ and $s'_b$ be the points in $H_1$ such that $s_a s_a'$ and $s_b s_b'$ are parallel to the $x_d$-axis. The line segment $s_a' s_b'$ lies in $H_1$, and is covered by quadtree cubes at level $\ell$. The orthogonal projections of these cubes to each coordinate axis is comprises consecutive intervals, hence we can find a Manhattan path between $s_a$ and $s_b$ along the edges of these cubes. Next assume that $s_a\in Q_a$ and $s_b\in Q_b$ are vertices of two cubes at different levels of $\mathcal{Q}_1$. If $Q_a$ and $Q_b$ are in ancestor-descendant relation, we can find a Manhattan path between $s_a$ and $s_b$ by tracing the edges created by the recursive subdivision. Otherwise, $Q_a$ and $Q_b$ are interior-disjoint. Assume that the side-length of $Q_a$ is less than that of $Q_b$; and let $Q'_b\subset Q_b$ be a quadtree cube on the same level as $Q_a$, and closest to $Q_a$. Then we can find a Manhattan path from $s_a$ to $s_b$ as a concatenation of two Manhattan paths via a vertex of $Q_b'$. Finally, for two vertices in $T_1$ and $T_2$, resp., we obtain a Manhattan path by concatenating two Manhattan paths via the origin, which is a vertex of both $T_1$ and $T_2$.

Since $T_1\cup T_2$ is a Manhattan network for  $(S_1\cup S_2)\cup (\widehat{S}_1\cup\widehat{S}_2)$, then $\opt\leq \|T_1\|+\|T_2\|=\Theta(\eps^{1-d})$.
Combined with $\alg\geq  \Omega(\eps^{1-2d})$, this yields
a lower bound of $\Omega(\eps^{1-2d}/\eps^{1-d})=\Omega(\eps^{-d})$ for the competitive ratio of any online algorithm in $\R^d$ under the $L_1$ norm.
The following theorem summarizes our results in dimensions $d\geq 3$.

\begin{theorem}
Let $d\geq 3$. For every $\eps>0$, the competitive ratio of any online algorithm for $(1+\eps)$-spanners in $\R^d$ under the $L_1$ norm is $\Omega(\eps^{-d})$.
\end{theorem}

\begin{remark}
Recall from the discussion in Section~\ref{sec:intro} that there are point sequences for which the insertion of a point may reduce the weight of the spanner (Fig.~\ref{fig:online-intro}). So, the weight of the optimum solution need not increase monotonically in the number of points. In this section, this phenomenon leads to a significant improvement over the weight of an optimum spanner under $L_1$ norm, and helped to design an improved lower bound for the competitive ratio. We do not know whether this phenomenon can be exploited to produce improved lower bounds under the $L_2$-norm.
%
\end{remark}

\section{Conclusions}
\label{sec:con}

We have studied online spanners for sequences of points in $\mathbb{R}^d$, in fixed dimensions $d\geq 1$, under $L_2$ and $L_1$ norms. We established a tight bound of $\Theta(\eps^{-1}\log n / \log \eps^{-1})$ for the competitive ratio of any online $(1+\eps)$-spanner algorithms on a real line (Theorem~\ref{thm:1D-bounds}).
However it remains an open problem to close the gap between the lower and upper bounds in $\R^d$, for $d\ge 2$. Under the $L_2$ norm, previously known algorithms achieve  competitive ratio $O(\eps^{-(d+1)}\log n)$ (Theorem~\ref{thm:L2woSteiner}).
The best lower bound we are aware of holds for $d=1$. It is unclear whether the lower bound can be improved to $\eps^{-\omega(d)}\log n$ for $d\geq 2$.

Next, we have
showed that, if an online algorithm is allowed to use Steiner points, it can achieve a substantially better competitive ratio in terms of $\eps$, namely $O(\eps^{(1-d)/2} \log n)$, for a sequence of $n$ points in $\R^d$ and any constant $d\ge 2$, under the $L_2$ norm (Theorem~\ref{thm:L2withSteiner}).
As a counterpart, we proved that any online spanner algorithm for a sequence of $n$ points in $\mathbb{R}^d$ under $L_2$ norm has competitive ratio $\Omega(f(n))$, where $\lim_{n\rightarrow \infty}f(n)=\infty$ (Theorem~\ref{thm:w-SP-LB}).
It remains an open problem whether the competitive ratio depends on $\eps$ for Euclidean Steiner spanners. Another open problem is whether the factor $\log n$ in the upper bounds can be reduced, e.g., to $\log n/\log \log n$; similar to the work by Alon and Azar~\cite{AlonA93} who established such a lower bound for Euclidean minimum Steiner trees (EMST) for $n$ points in $\R^2$.

We have established a lower bound  $\Omega(\eps^{-d})$ for the competitive ratio under the $L_1$-norm in $\mathbb{R}^d$. It is unclear whether it can be improved by a $\log n$ factor in dimensions $d\geq 2$. Designing online algorithms that match these bounds under the $L_1$ norm is left for future research.

In online spanner algorithms, the decisions are irrevocable, which means that once an edge is added to the spanner by an online algorithm, it can never be deleted. However, if some of the decisions are reversible, better bounds may be possible. This model is commonly known as \emph{online algorithms with recourse}~\cite{gu2016power, imase1991dynamic, megow2016power}. In $1$-dimension, for instance, an optimum spanner is just a monotone path connecting the points in linear order, and any online algorithm that is allowed to remove at least one edge at per iteration can maintain such a path. In higher dimensions, however, it is unclear whether a $O(1)$-approximation of the minimum-weight $(1+\eps)$-spanner can be maintained with $O(\eps^{-d+1})$ recourse.



\bibliography{online}

\begin{thebibliography}{10}

\bibitem{alon2006general}
Noga Alon, Baruch Awerbuch, Yossi Azar, Niv Buchbinder, and Joseph Naor.
\newblock A general approach to online network optimization problems.
\newblock {\em ACM Transactions on Algorithms ({TALG})}, 2(4):640--660, 2006.

\bibitem{AlonA93}
Noga Alon and Yossi Azar.
\newblock On-line {S}teiner trees in the {E}uclidean plane.
\newblock {\em Discrete \& Computational Geometry}, 10:113--121, 1993.

\bibitem{althofer1993sparse}
Ingo Alth{\"o}fer, Gautam Das, David Dobkin, Deborah Joseph, and Jos{\'e}
  Soares.
\newblock On sparse spanners of weighted graphs.
\newblock {\em Discrete \& Computational Geometry}, 9(1):81--100, 1993.

\bibitem{arya1994randomized}
Sunil Arya, David~M Mount, and Michiel Smid.
\newblock Randomized and deterministic algorithms for geometric spanners of
  small diameter.
\newblock In {\em Proc.\ 35th IEEE Symposium on Foundations of Computer Science
  ({FOCS})}, pages 703--712, 1994.

\bibitem{arya1997efficient}
Sunil Arya and Michiel Smid.
\newblock Efficient construction of a bounded-degree spanner with low weight.
\newblock {\em Algorithmica}, 17(1):33--54, 1997.

\bibitem{awerbuch2004line}
Baruch Awerbuch, Yossi Azar, and Yair Bartal.
\newblock On-line generalized {S}teiner problem.
\newblock {\em Theoretical Computer Science}, 324(2-3):313--324, 2004.

\bibitem{AwerbuchBP90}
Baruch Awerbuch, Alan~E. Baratz, and David Peleg.
\newblock Cost-sensitive analysis of communication protocols.
\newblock In {\em Proc. 9th {ACM} Symposium on Principles of Distributed
  Computing ({PODC})}, pages 177--187, 1990.

\bibitem{BaswanaKS12}
Surender Baswana, Sumeet Khurana, and Soumojit Sarkar.
\newblock Fully dynamic randomized algorithms for graph spanners.
\newblock {\em {ACM} Trans. Algorithms}, 8(4):35:1--35:51, 2012.

\bibitem{BergamaschiHGWW21}
Thiago Bergamaschi, Monika Henzinger, Maximilian~Probst Gutenberg,
  Virginia~Vassilevska Williams, and Nicole Wein.
\newblock New techniques and fine-grained hardness for dynamic near-additive
  spanners.
\newblock In {\em Proc. {ACM-SIAM} Symposium on Discrete Algorithms ({SODA})},
  pages 1836--1855, 2021.

\bibitem{berman1997line}
Piotr Berman and Chris Coulston.
\newblock On-line algorithms for steiner tree problems.
\newblock In {\em Proc. 29th ACM Symposium on Theory of Computing ({STOC})},
  pages 344--353, 1997.

\bibitem{BernsteinFH19}
Aaron Bernstein, Sebastian Forster, and Monika Henzinger.
\newblock A deamortization approach for dynamic spanner and dynamic maximal
  matching.
\newblock In {\em Proc.\ 13th {ACM-SIAM} Symposium on Discrete Algorithms
  ({SODA})}, pages 1899--1918, 2019.

\bibitem{BT-lessp-21}
Sujoy Bhore and Csaba~D. T{\'{o}}th.
\newblock Light {E}uclidean {S}teiner spanners in the plane.
\newblock In {\em Proc. 37th International Symposium on Computational Geometry
  ({SoCG})}, volume 189 of {\em LIPIcs}, pages 31:1--17. Schloss Dagstuhl,
  2021.

\bibitem{BT-oess-21}
Sujoy Bhore and Csaba~D. T{\'{o}}th.
\newblock On {E}uclidean {S}teiner (1+{\(\varepsilon\)})-spanners.
\newblock In {\em Proc. 38th Symposium on Theoretical Aspects of Computer
  Science ({STACS})}, volume 187 of {\em LIPIcs}, pages 13:1--13:16. Schloss
  Dagstuhl, 2021.

\bibitem{BY98}
Allan Borodin and Ran El{-}Yaniv.
\newblock {\em Online computation and competitive analysis}.
\newblock Cambridge University Press, 1998.

\bibitem{bose2004ordered}
Prosenjit Bose, Joachim Gudmundsson, and Pat Morin.
\newblock Ordered theta graphs.
\newblock {\em Computational Geometry}, 28(1):11--18, 2004.

\bibitem{BoseS13}
Prosenjit Bose and Michiel H.~M. Smid.
\newblock On plane geometric spanners: {A} survey and open problems.
\newblock {\em Comput. Geom.}, 46(7):818--830, 2013.

\bibitem{callahan1993optimal}
Paul~B. Callahan.
\newblock Optimal parallel all-nearest-neighbors using the well-separated pair
  decomposition.
\newblock In {\em Proc.\ 34th IEEE Symposium on Foundations of Computer Science
  ({FOCS})}, pages 332--340, 1993.

\bibitem{CallahanK93}
Paul~B. Callahan and S.~Rao Kosaraju.
\newblock Faster algorithms for some geometric graph problems in higher
  dimensions.
\newblock In Vijaya Ramachandran, editor, {\em Proc. 4th {ACM-SIAM} Symposium
  on Discrete Algorithms ({SODA})}, pages 291--300, 1993.

\bibitem{CallahanK95}
Paul~B. Callahan and S.~Rao Kosaraju.
\newblock A decomposition of multidimensional point sets with applications to
  $k$-nearest-neighbors and $n$-body potential fields.
\newblock {\em J. {ACM}}, 42(1):67--90, 1995.

\bibitem{carmi2013minimum}
Paz Carmi and Lilach Chaitman-Yerushalmi.
\newblock Minimum weight {E}uclidean $t$-spanner is {NP}-hard.
\newblock {\em Journal of Discrete Algorithms}, 22:30--42, 2013.

\bibitem{ChanHJ20}
Timothy~M. Chan, Sariel Har{-}Peled, and Mitchell Jones.
\newblock On locality-sensitive orderings and their applications.
\newblock {\em {SIAM} J. Comput.}, 49(3):583--600, 2020.

\bibitem{Chew86}
L.~Paul Chew.
\newblock There is a planar graph almost as good as the complete graph.
\newblock In {\em Proc.\ 2nd Symposium on Computational Geometry ({SoCG})},
  pages 169--177. {ACM} Press, 1986.

\bibitem{Chew89}
L.~Paul Chew.
\newblock There are planar graphs almost as good as the complete graph.
\newblock {\em J. Comput. Syst. Sci.}, 39(2):205--219, 1989.

\bibitem{Clarkson87}
Kenneth~L. Clarkson.
\newblock Approximation algorithms for shortest path motion planning.
\newblock In {\em Proc.\ 19th {ACM} Symposium on Theory of Computing {(STOC)}},
  pages 56--65, 1987.

\bibitem{das1993optimally}
Gautam Das, Paul Heffernan, and Giri Narasimhan.
\newblock Optimally sparse spanners in 3-dimensional {E}uclidean space.
\newblock In {\em Proc. 9th Symposium on Computational Geometry ({SoCG})},
  pages 53--62. ACM Press, 1993.

\bibitem{narasimhan1995new}
Gautam Das, Giri Narasimhan, and Jeffrey~S. Salowe.
\newblock A new way to weigh malnourished {E}uclidean graphs.
\newblock In {\em Proc.\ 6th ACM-SIAM Symposium on Discrete Algorithms
  {(SODA)}}, pages 215--222, 1995.

\bibitem{BergCKO08}
Mark de~Berg, Otfried Cheong, Marc~J. van Kreveld, and Mark~H. Overmars.
\newblock {\em Computational Geometry: {A}lgorithms and Applications}.
\newblock Springer, 3 edition, 2008.

\bibitem{elkin2015steiner}
Michael Elkin and Shay Solomon.
\newblock Steiner shallow-light trees are exponentially lighter than spanning
  ones.
\newblock {\em SIAM Journal on Computing}, 44(4):996--1025, 2015.

\bibitem{FischerH05}
John Fischer and Sariel Har{-}Peled.
\newblock Dynamic well-separated pair decomposition made easy.
\newblock In {\em Proc. 17th Canadian Conference on Computational Geometry
  ({CCCG})}, pages 235--238, 2005.

\bibitem{gao2006deformable}
Jie Gao, Leonidas~J. Guibas, and An~Nguyen.
\newblock Deformable spanners and applications.
\newblock {\em Comput. Geom.}, 35(1-2):2--19, 2006.

\bibitem{gottlieb2017efficient}
Lee-Ad Gottlieb, Aryeh Kontorovich, and Robert Krauthgamer.
\newblock Efficient regression in metric spaces via approximate {L}ipschitz
  extension.
\newblock {\em IEEE Transactions on Information Theory}, 63(8):4838--4849,
  2017.

\bibitem{gottlieb2008optimal}
Lee-Ad Gottlieb and Liam Roditty.
\newblock An optimal dynamic spanner for doubling metric spaces.
\newblock In {\em Proc. 16th European Symposium on Algorithms ({ESA})}, volume
  5193 of {\em LNCS}, pages 478--489. Springer, 2008.

\bibitem{gu2016power}
Albert Gu, Anupam Gupta, and Amit Kumar.
\newblock The power of deferral: {M}aintaining a constant-competitive {S}teiner
  tree online.
\newblock {\em SIAM Journal on Computing}, 45(1):1--28, 2016.

\bibitem{GudmundssonK18}
Joachim Gudmundsson and Christian Knauer.
\newblock Dilation and detours in geometric networks.
\newblock In Teofilo~F. Gonzalez, editor, {\em Handbook of Approximation
  Algorithms and Metaheuristics}, volume~2. Chapman and Hall/CRC, 2nd edition,
  2018.

\bibitem{GudmundssonLN02}
Joachim Gudmundsson, Christos Levcopoulos, and Giri Narasimhan.
\newblock Fast greedy algorithms for constructing sparse geometric spanners.
\newblock {\em {SIAM} J. Comput.}, 31(5):1479--1500, 2002.

\bibitem{gudmundsson2008approximate}
Joachim Gudmundsson, Christos Levcopoulos, Giri Narasimhan, and Michiel Smid.
\newblock Approximate distance oracles for geometric spanners.
\newblock {\em ACM Transactions on Algorithms ({TALG})}, 4(1):1--34, 2008.

\bibitem{hajiaghayi2013online}
Mohammad~Taghi Hajiaghayi, Vahid Liaghat, and Debmalya Panigrahi.
\newblock Online node-weighted {S}teiner forest and extensions via disk
  paintings.
\newblock In {\em Porc. 54th IEEE Symposium on Foundations of Computer Science
  ({FOCS})}, pages 558--567, 2013.

\bibitem{Sariel-Book}
Sariel Har{-}Peled.
\newblock {\em Geometric Approximation Algorithms}, volume 173 of {\em
  Mathematical Surveys and Monographs}.
\newblock AMS, Providence, RI, 2011.

\bibitem{imase1991dynamic}
Makoto Imase and Bernard~M. Waxman.
\newblock Dynamic {S}teiner tree problem.
\newblock {\em SIAM Journal on Discrete Mathematics}, 4(3):369--384, 1991.

\bibitem{keil1988approximating}
J.~Mark Keil.
\newblock Approximating the complete {E}uclidean graph.
\newblock In {\em Proc.\ 1st Scandinavian Workshop on Algorithm Theory
  ({SWAT})}, volume 318 of {\em LNCS}, pages 208--213. Springer, 1988.

\bibitem{KhullerRY93}
Samir Khuller, Balaji Raghavachari, and Neal~E. Young.
\newblock Balancing minimum spanning and shortest path trees.
\newblock In {\em Proc. 4th {ACM-SIAM} Symposium on Discrete Algorithms
  ({SODA})}, pages 243--250, 1993.

\bibitem{le2019truly}
Hung Le and Shay Solomon.
\newblock Truly optimal {E}uclidean spanners.
\newblock In {\em Proc.\ 60th IEEE Symposium on Foundations of Computer Science
  ({FOCS})}, pages 1078--1100, 2019.

\bibitem{le2020light}
Hung Le and Shay Solomon.
\newblock Light {E}uclidean spanners with {S}teiner points.
\newblock In {\em Proc.\ 28th European Symposium on Algorithms ({ESA})}, volume
  173 of {\em LIPIcs}, pages 67:1--67:22. Schloss Dagstuhl, 2020.

\bibitem{le2020unified}
Hung Le and Shay Solomon.
\newblock A unified and fine-grained approach for light spanners.
\newblock {\em CoRR}, abs/2008.10582, 2020.
\newblock \href {http://arxiv.org/abs/2008.10582} {\path{arXiv:2008.10582}}.

\bibitem{megow2016power}
Nicole Megow, Martin Skutella, Jos{\'e} Verschae, and Andreas Wiese.
\newblock The power of recourse for online {MST} and {TSP}.
\newblock {\em SIAM Journal on Computing}, 45(3):859--880, 2016.

\bibitem{naor2011online}
Joseph Naor, Debmalya Panigrahi, and Mohit Singh.
\newblock Online node-weighted {S}teiner tree and related problems.
\newblock In {\em Proc. 52nd IEEE Symposium on Foundations of Computer Science
  {(FOCS)}}, pages 210--219, 2011.

\bibitem{narasimhan2007geometric}
Giri Narasimhan and Michiel Smid.
\newblock {\em Geometric Spanner Networks}.
\newblock Cambridge University Press, 2007.

\bibitem{rao1998approximating}
Satish~B. Rao and Warren~D. Smith.
\newblock Approximating geometrical graphs via “spanners” and
  “banyans”.
\newblock In {\em Proc.\ 13th ACM Symposium on Theory of Computing {(STOC)}},
  pages 540--550, 1998.

\bibitem{Roditty12}
Liam Roditty.
\newblock Fully dynamic geometric spanners.
\newblock {\em Algorithmica}, 62(3-4):1073--1087, 2012.

\bibitem{schindelhauer2007geometric}
Christian Schindelhauer, Klaus Volbert, and Martin Ziegler.
\newblock Geometric spanners with applications in wireless networks.
\newblock {\em Comput. Geom.}, 36(3):197--214, 2007.

\bibitem{Smid18}
Michiel H.~M. Smid.
\newblock The well-separated pair decomposition and its applications.
\newblock In Teofilo~F. Gonzalez, editor, {\em Handbook of Approximation
  Algorithms and Metaheuristics}, volume~2. Chapman and Hall/CRC, 2nd edition,
  2018.

\bibitem{Solomon15}
Shay Solomon.
\newblock Euclidean {S}teiner shallow-light trees.
\newblock {\em J. Comput. Geom.}, 6(2):113--139, 2015.

\bibitem{yao1982constructing}
Andrew~Chi{-}Chih Yao.
\newblock On constructing minimum spanning trees in $k$-dimensional spaces and
  related problems.
\newblock {\em {SIAM} J. Comput.}, 11(4):721--736, 1982.

\end{thebibliography}


\end{document}